\documentclass[a4paper,11pt,fleqn]{article}

\usepackage{amsmath,amssymb,amsthm}
\usepackage[pdftex,paper=a4paper,left=3.0cm,right=3.0cm,top=4cm,bottom=4cm]{geometry}
\usepackage{graphicx}
\usepackage{subfig}

\usepackage{xspace}

\newlength{\addressminus}
\newlength{\addressminusB}
\title{Smoothed Performance Guarantees for Local Search \footnote{A preliminary version of this paper appeared in the proceedings of ESA 2011.} }

\author{Tobias Brunsch\\[\addressminusB]
        \footnotesize Dept.~of Computer Science\\[\addressminus]
      	\footnotesize University of Bonn, Germany\\[\addressminus]
        \footnotesize {\tt brunsch@cs.uni-bonn.de}
 \and Heiko R\"oglin\\[\addressminusB]
        \footnotesize Dept.~of Computer Science\\[\addressminus]
      	\footnotesize University of Bonn, Germany\\[\addressminus]
        \footnotesize {\tt heiko@roeglin.org}
 \and Cyriel Rutten\\[\addressminusB]
 	\footnotesize Dept.~of Quantitative Economics\\[\addressminus]
  	\footnotesize Maastricht University, The Netherlands\\[\addressminus]
  	\footnotesize {\tt cyrielrutten@gmail.com}
 \and Tjark Vredeveld\\[\addressminusB]
	\footnotesize Dept.~of Quantitative Economics\\[\addressminus]
  	\footnotesize Maastricht University, The Netherlands\\[\addressminus]
  	\footnotesize {\tt t.vredeveld@maastrichtuniversity.nl}
}

\date{}

\begin{document}

\maketitle

\allowdisplaybreaks

\newtheorem{theorem}{Theorem}
\newtheorem{lemma}[theorem]{Lemma}
\newtheorem{cor}[theorem]{Corollary}
\newtheorem{prop}[theorem]{Proposition}
\newtheorem{claim}{Claim}

\newtheorem{definition}[theorem]{Definition}

\newtheorem*{remark}{Remark}
\newtheorem{property}{Property}

\newenvironment{artclfig}{\begin{figure}[htbp]\begin{center}}{\end{center}\end{figure}}

\newcommand{\sched}{\ensuremath{\sigma}\xspace}
\newcommand{\cmax}{\ensuremath{C_{\max}}\xspace}
\newcommand{\copt}{\ensuremath{\cmax^*}\xspace}
\newcommand{\csched}[1][\sched]{\ensuremath{\cmax(#1)}\xspace}
\newcommand{\load}[2][]{\ensuremath{L#1_{#2}}\xspace}

\newcommand{\Prob}{\ensuremath{\mathop{\mathbf{Pr}}}\xspace}
\newcommand{\Probe}[1]{\ensuremath{\Prob[#1]\,}\xspace}
\newcommand{\Probl}[2][]{\ensuremath{\Prob\limits_{#1}\left[#2\right]}\xspace}

\newcommand{\E}{\ensuremath{\mathop{\mathbf{E}}}\xspace}
\newcommand{\Ee}[1]{\ensuremath{\E[#1]}\xspace}
\newcommand{\El}[2][]{\ensuremath{\E\limits_{#1}\left[#2\right]}\xspace}

\newcommand{\opt}{\ensuremath{\sched^*}\xspace}
\newcommand{\critmach}{\ensuremath{i_{\max}}\xspace}
\newcommand{\loadsched}[1]{\ensuremath{L_{#1}}(\sched)\xspace}
\newcommand{\jobset}[1][]{\ensuremath{J_{#1}}\xspace}
\newcommand{\jobsetsched}[1]{\ensuremath{J_{#1}}(\sched)\xspace}
\newcommand{\pmax}{\ensuremath{p_{\max}}\xspace}
\newcommand{\event}{\ensuremath{\mathcal{E}}\xspace}
\newcommand{\e}{\ensuremath{e}\xspace}
\newcommand{\allow}[1]{\ensuremath{{\cal M}_{#1}}\xspace}
\newcommand{\psmooth}{\ensuremath{\mathbf{p}}\xspace}

\newcommand{\jobs}[2][\sched]{\ensuremath{J_{#2}({#1})}\xspace}
\newcommand{\schedopt}{\sched^*}

\newcommand{\jobclass}[1]{\ensuremath{\mathcal{J}_{#1}}}

\newcommand{\smin}{s_{\min}}
\newcommand{\smax}{s_{\max}}

\newcommand{\JUMP}[1]{\mathrm{Jump}(#1)}
\newcommand{\LEX}[1]{\mathrm{Lex}(#1)}
\newcommand{\LIST}[1]{\mathrm{List}(#1)}
\newcommand{\NL}[1]{\mathrm{NL}(#1)}

\newcommand{\SET}[1]{\left\{#1\right\}}
\newcommand{\MIN}[1]{\min \SET{#1}}
\newcommand{\MAX}[1]{\max \SET{#1}}
\renewcommand{\d}{\mathrm{d}}

\newcommand{\smalljobs}{{\jobset}_{\text{small}}}
\renewcommand{\H}[1]{R_{#1}}

\newcommand{\WHERE}{\,\colon\,}
\newcommand{\FLOOR}[1]{\left\lfloor#1\right\rfloor}
\newcommand{\CEIL}[1]{\left\lceil#1\right\rceil}
\newcommand{\COMMA}{\,,}
\newcommand{\DOT}{\,.}

\newcommand{\ProofText}[1]{Proof of #1}

\providecommand{\qedhere}{\tag*{\qed}}

\begin{abstract}

We study popular local search and greedy algorithms for standard machine
scheduling problems.
The performance guarantee of these algorithms is well understood,
but the worst-case lower bounds seem somewhat contrived and it is
questionable whether they arise in practical applications. To find out how
robust these bounds are, we study the algorithms in the framework  
of smoothed analysis, in which instances are subject to some
degree of random noise.

While the lower bounds for all scheduling variants with restricted machines
are rather robust, we find out that the bounds are fragile for unrestricted
machines. In particular, we show that the smoothed performance guarantee
of the jump and the lex-jump algorithm are (in contrast to the worst case)
independent of the number of machines. They are~$\Theta(\phi)$ and~$\Theta(\log \phi)$,
respectively, where~$1/\phi$ is a parameter measuring the magnitude 
of the perturbation. The latter immediately implies that also the smoothed
price of anarchy is~$\Theta(\log \phi)$ for routing games on parallel links.
Additionally, we show that for unrestricted machines
also the greedy list scheduling algorithm has an approximation guarantee
of~$\Theta(\log \phi)$.

\end{abstract}

\section{Introduction}
\label{sec:intro}

The performance guarantee of local search and greedy algorithms for scheduling
problems is well studied and understood. For most algorithms, matching upper and
lower bounds on their approximation ratio are known. The lower bounds are often
somewhat contrived, however, and it is questionable whether they resemble typical
instances in practical applications. For that reason, we study these algorithms
in the framework of smoothed analysis, in which instances are subject to some
degree of random noise. By doing so, we find out for which heuristics and
scheduling variants the lower bounds are robust and for which they are fragile and
not very likely to occur in practical applications. Since pure Nash equilibria
can be seen as local optima, our results also imply a
new bound on the smoothed price of anarchy, showing that known worst-case
results are too pessimistic in the presence of noise.

Let us first describe the scheduling problems that we study. We assume that
there is a set $\jobset = \SET{ 1, \ldots, n }$ of jobs each of which needs to
be processed on one of the machines from the set $M = \SET{ 1, \ldots, m }$.
All jobs and machines are available for processing at time~$0$. The goal is to
schedule the jobs on the machines such that the \emph{makespan}, i.e., the time
at which the last job is completed, is minimized. Each machine~$i \in M$ has a
speed~$s_i$ and each job~$j \in \jobset$ has a processing
requirement~$p_j$. 
The time~$p_{ij}$ it takes to fully process job~$j$ on machine~$i$ depends on
the machine environment. We consider two machine environments. The first one is
the one of \emph{uniform parallel machines}, also known as
\emph{related machines}:
$p_{ij} = p_j / s_i$. The second machine environment that we consider is the
one of \emph{restricted related machines}: a job~$j$ is only allowed to be
processed on a subset $\allow{j} \subseteq M$ of the machines. The processing
time is therefore $p_{ij} = p_j / s_i$ if $i \in \allow{j}$ and $p_{ij} =
\infty$ if $i \notin \allow{j}$. An instance~$I$ of a scheduling problem
consists of the machine speeds $s_1, \ldots, s_m$, the processing requirements
$p_1, \ldots, p_n$, and in the restricted case the allowed machine
set~$\allow{j} \subseteq M$ for every job~$j$.

A special case for both machine environments is when all speeds are
equal, i.e., $s_i = 1$ for all $i \in M$. In this case, we say that the
machines are identical. In the notation of Graham et
al.~\cite{graham:etal:79} these problems are denoted by $Q||\cmax$ and
$Q|\allow{j}|\cmax$ for the related machine problems and $P||\cmax$ and
$P|\allow{j}|\cmax$ in case of identical machines.
In these problems, makespan minimization is equivalent to
minimizing the maximum machine finishing time. Once the assignment of
the jobs to the machines is known, the order in which the jobs are
processed is of no importance to determine the machine finishing times,
as long as the jobs are processed without any idle time in between.
Therefore, we assume that the
jobs that are scheduled on a machine $i$ share this processor in
such a way that they all finish at the same time.

Even in the case that all speeds are equal, the problems under
 consideration are known to be
strongly NP-hard when $m$ is part of the input (see, e.g., Garey and
Johnson~\cite{Garey+Johnson:1979}). This has motivated a lot of research in the
previous decades on approximation algorithms for scheduling problems. Since some
of the theoretically best approximation algorithms are rather involved, a lot of
research has focused on simple heuristics like \emph{greedy algorithms} and \emph{local
search algorithms} which are easy to implement. While greedy algorithms
make reasonable ad hoc decisions to obtain a schedule, local search algorithms start with
some schedule and iteratively improve the current schedule by performing some
kind of local improvements until no such is possible anymore. 
In this article, we consider the following algorithms that can be applied to all scheduling variants that we have described above:
\begin{itemize}

  \item \emph{List scheduling} is a greedy algorithm that starts from an empty schedule and a list of jobs. Then, it repeatedly selects the next unscheduled job from the list and assigns it to the machine on which it will be completed the earliest with respect to the current partial schedule. We call any schedule that can be generated by list scheduling a \emph{list schedule}.

  \item The \emph{jump} and the \emph{lex-jump} algorithms are local
search algorithms that start with an arbitrary schedule and iteratively
perform a local improvement step. In each improvement step, one job is
reassigned (jumped) from a machine~$i$ to a different machine~$i'$ where it
finishes earlier. In the jump algorithm, only jobs on \emph{critical}
machines~$i$, i.e., machines that have maximum finishing time, are
considered to be improving. In the lex-jump algorithm, the jobs can
be arbitrary. Note that a local step is lex-jump improving if and only if the sorted vector of machine finishing times decreases lexicographically, hence the term lex-jump.
A schedule for which there is no jump improvement step or no lex-jump improvement step is called \emph{jump optimal} or \emph{lex-jump optimal}, respectively.

\end{itemize}
For each of these three algorithms, we are interested in  their performance guarantees, i.e., the worst case bound on the ratio of the makespan of a schedule to be returned by the algorithm over the makespan of an optimal schedule.
The final schedule returned by a local search algorithm is called a \emph{local optimum}.
Usually, there are multiple local optima for a given scheduling
instance both for the jump and the lex-jump algorithm with varying quality.
As we do not know which local optimum is found by the local search, we 
will always bound the quality of the worst local optimum. Since local optima
for lex-jump and pure Nash equilibria are the same, see e.g.~\cite{Voecking:2007:AGT}, this corresponds
to bounding the price of anarchy in the scheduling game that is obtained if 
jobs are selfish agents trying to minimize their own completion time and if the
makespan is considered as the welfare function.
Similarly, list scheduling can produce different schedules
depending on the order in which the jobs are inserted into the list. Also for
list scheduling we will bound the quality of the worst schedule that can be
obtained.

\paragraph{Notation.}

Consider an
instance~$I$ for the scheduling problem and a schedule~$\sched$ for this
instance. By~$\jobset[i](\sched) \subseteq \jobset$ we denote the set of jobs
assigned to machine~$i$ according to~$\sched$. The \emph{processing requirement
on a machine~$i \in M$} is defined as $\sum_{j \in \jobset[i](\sched)} p_j$
and the \emph{load} of a machine is defined by $\load{i}(I,\sched) =
\sum_{j \in \jobset[i](\sched)} p_{ij}$.
The makespan~$\csched[I, \sched]$ of~$\sched$ can be written as $\csched[I,
\sched] = \max_{i \in M} \load{i}(I, \sched)$. The optimal makespan, i.e., the
makespan of an optimal schedule is denoted by~$\copt(I)$. By~$\JUMP{I}$,
$\LEX{I}$, and~$\LIST{I}$ we denote the set of all feasible jump optimal
schedules, lex-jump optimal schedules, and list schedules, respectively,
according to instance~$I$.

If the instance~$I$ is clear from the context, we simply
write~$\load{i}(\sched)$ instead of~$\load{i}(I, \sched)$, $\csched[\sigma]$
instead of $\csched[I, \sched]$, and~$\copt$ instead of~$\copt(I)$. If the
schedule~$\sched$ is clear as well, we simplify our notation further
to~$\load{i}$ and~$\cmax$ and we write~$\jobset[i]$ instead
of~$\jobset[i](\sched)$. By appropriate scaling, we may assume w.l.o.g.\
that the slowest machine has speed~$\smin = 1$ and that all processing 
requirements are bounded by $p_j \leq 1$.
In Appendix~\ref{sec:appendix-table}, the notation is summarized in a table.

\paragraph{Smoothed analysis.} 
As can be seen in Table~\ref{tab:results}, the worst-case
approximation guarantee of jump and lex-jump is known for all scheduling
variants and it is constant only for the simplest case with unrestricted
and identical machines. In all other cases it increases with the number~$m$
of machines. For list scheduling, the case with unrestricted and
related machines has been considered. Cho and Sahni~\cite{Cho+Sahni:1980} and Aspnes et al.~\cite{DBLP:journals/jacm/AspnesAFPW97} showed that the performance guarantee of list scheduling is $\Theta(\log m)$ in this case.

In order to analyze the robustness of the worst-case bounds, we turn to the
framework of smoothed analysis, introduced by Spielman and
Teng~\cite{Spielman+Teng:SA:2004} to explain why certain algorithms perform
well in practice in spite of a poor worst-case running time. Smoothed analysis is a hybrid of
average-case and worst-case analysis: First, an adversary chooses an instance.
Second, this instance is slightly randomly perturbed. The smoothed performance
is the expected performance, where the expectation is taken over the random
perturbation. The adversary, trying to make the algorithm perform as bad as
possible, chooses an instance that maximizes this expected performance. This
assumption is made to model that often the input an algorithm gets is subject to
imprecise measurements, rounding errors, or numerical imprecision.
If the smoothed performance guarantee of an algorithm is small, then bad worst-case
instances might exist, but one is very unlikely to encounter them if instances are
subject to some small amount of random noise. 

We follow the more general model of smoothed analysis introduced by Beier and
V\"{o}cking~\cite{DBLP:journals/jcss/BeierV04}.
In this model, the adversary is even allowed to
specify the probability distribution of the random noise. The influence he can
exert is described by a parameter~$\phi \geq 1$ denoting the maximum density of the noise. This model is formally defined as follows.
\begin{definition}
\label{def:phi-smooth}
In a \emph{$\phi$-smooth} instance~$\mathcal{I}$,
the adversary chooses the following input data:
\begin{itemize}
\item the number~$m$ of machines;
\item arbitrary machine speeds $\smax := s_1 \geq \ldots \geq s_m =: \smin = 1$,
in the case of non-identical machines;
\item the number~$n$ of jobs;
\item an arbitrary set $\allow{j} \subseteq M$ for each job~$j\in \jobset$,
in the case of restricted machines;
\item for each~$p_j$, a probability density $f_j:[0,1]\to[0,\phi]$ according
to which $p_j$ is chosen independently of the processing requirements of the other jobs.
\end{itemize}
Note that the only perturbed part of the instance are the processing
requirements.
Formally, a $\phi$-smooth instance is not a single instance but a
distribution over instances.
We write $I \sim \mathcal{I}$ to denote that the instance $I$ is
drawn from the $\phi$-smooth instance $\mathcal{I}$.
\end{definition}

The parameter~$\phi$ specifies how close the analysis is to a worst case analysis. The adversary can, for example, choose for every~$p_j$ an interval of length~$1/\phi$ from which $p_j$ is drawn uniformly at random. For~$\phi = 1$, every processing requirement is uniformly distributed over~$[0, 1]$, and hence the input model equals the average case for uniformly distributed processing times. When~$\phi$ gets larger, the adversary can specify the processing requirements more and more precisely, and for~$\phi \to \infty$ the smoothed analysis approaches a worst-case analysis.
 
In this article, we analyze the \emph{smoothed performance guarantee}
of the jump, the lex-jump, and the list scheduling algorithm. As mentioned above, 
to define the approximation guarantee of these algorithms on a given instance, we
consider the worst local optimum (for the jump and the lex-jump algorithm) or the worst order in which the jobs are inserted into the list (for the list scheduling algorithm). 
Now, the smoothed performance is defined to be the worst expected approximation guarantee of any
$\phi$-smooth instance.  

\paragraph{Our results.}

\begin{table}
\begin{center}\footnotesize
\begin{tabular}{|l|cc|cc|}\hline
& \multicolumn{2}{c|}{worst case} & \multicolumn{2}{c|}{$\phi$-smooth}\\
& jump & lex-jump & jump & lex-jump \\
\hline

\begin{tabular}{l}unrestricted\\identical\end{tabular} 
& $\Theta(1)$ \cite{finn:horowitz:1979,Schuurman:Vredeveld:2007} 
& $\Theta(1)$ \cite{finn:horowitz:1979,Schuurman:Vredeveld:2007} 
& $\Theta(1)$ 
& $\Theta(1)$\\

\begin{tabular}{l}unrestricted\\related\end{tabular}
& $\Theta \left( \sqrt{m} \right)$ \cite{Cho+Sahni:1980,Schuurman:Vredeveld:2007}
& $\Theta \left(  \frac{\log m}{\log \log m} \right)$ \cite{Voecking:2007}
& $\Theta(\phi)$ [\ref{subsec:jump}]
& $\Theta(\log \phi)$ [\ref{subsec:ub-list-lex-jump}, \ref{subsec:lb-list-lex-jump}]\\

\begin{tabular}{l}restricted\\identical\end{tabular}
& $\Theta \left( \sqrt{m} \right)$ \cite{Rutten:etal:2012}
& $\Theta \left( \frac{\log m}{\log\log m} \right)$ \cite{Awerbuch+etal:2006}
& $\Theta \left( \sqrt{m} \right)$ [\ref{subsec:lb-jump-restricted}]
& $\Theta \left( \frac{\log m}{\log\log m} \right)$ [\ref{subsec:lb-lex-jump-restricted}]\\

\begin{tabular}{l}restricted\\related\end{tabular}
& $\Theta\Big(\sqrt{m\cdot{\smax}}\Big)$ \cite{Rutten:etal:2012}
& $\Theta\Big(\frac{\log S}{\log \log S}\Big)$ \cite{Rutten:etal:2012}
& $\Theta\Big(\sqrt{m\cdot {\smax}}\Big)$ [\ref{subsec:lb-jump-restricted}]
& $\Omega\Big(\frac{\log m}{\log \log m}\Big)$ [\ref{subsec:lb-lex-jump-restricted}]\\
\hline
\end{tabular}
\vspace{0.2cm}
\caption{Worst-case and smoothed performance guarantees for jump and lex-jump optimal
schedules. Here, $S = \sum_{i=1}^m s_i$, and we assume w.l.og.~that $\smin =1$. With [X.Y] we refer to 
the section in this article where the bound is shown.
\label{tab:results}}
\end{center}
\vspace{-0.8cm}
\end{table}

Our results for the jump and lex-jump algorithm are summarized in
Table~\ref{tab:results}. The first
remarkable observation is that the smoothed performance guarantees for all
variants of restricted machines are robust against random noise. We show that
even for large perturbations with constant~$\phi$, the worst-case lower bounds
carry over. This can be seen as an indication that neither the jump algorithm nor the lex-jump algorithm
yield a good approximation ratio for scheduling with restricted machines in practice.

The situation is much more promising for the unrestricted variants.
Here, the worst-case bounds are fragile and do not carry over to the smoothed
case. The interesting case is the one of unrestricted and related machines. Even though
both for jump and for lex-jump the worst-case lower bound is not robust,
there is a significant difference between these two: while the smoothed
approximation ratio for jump grows linearly with the perturbation parameter~$\phi$,
it grows only logarithmically in~$\phi$ for lex-jump optimal schedules.
This proves that also in the presence of random noise lex-jump optimal schedules
are significantly better than jump optimal schedules. As mentioned earlier, this also implies
that the smoothed price of anarchy is~$\Theta(\log \phi)$. Additionally, we
show that the smoothed approximation ratio of list scheduling
is~$\Theta(\log \phi)$ as well, even when the order of the list may be
specified after the realizations of the processing times are known.
 This indicates that both the lex-jump algorithm and the list scheduling algorithm should yield good
approximations on practical instances.

\paragraph{Related work.}
The approximability of $Q||\cmax$ is well understood. Cho and
Sahni~\cite{Cho+Sahni:1980} showed that list scheduling has a
performance guarantee of at most $1+\sqrt{2m-2}/2$ for~$m \geq 3$ and
that it is at least $\Omega(\log m)$. Aspnes et al.~\cite{DBLP:journals/jacm/AspnesAFPW97} improved the upper bound to $O(\log m)$ matching the lower bound asymptotically. Hochbaum and
Shmoys~\cite{Hochbaum+Shmoys:1988} designed a polynomial time
approximation scheme for this problem. Polynomial time approximation algorithms and polynomial time
approximation schemes for special cases of the problem on restricted
related machines are given in, among others,~\cite{Li06,Glass07,Ou08}.
More work on restricted related parallel machines is discussed in the
survey of Leung and Li~\cite{Leung08}.

In the last decade, there has been a strong interest in understanding
the worst-case behavior of local optima. We
refer to the survey~\cite{Angel:2006} and the
book~\cite{Michiels:etal:2007} for a comprehensive overview of the worst-case analysis
and other theoretical aspects of local search. 
It follows from the work of Cho and Sahni~\cite{Cho+Sahni:1980} that for
the problem on unrestricted related machines the performance guarantee
of the jump algorithm is $(1+\sqrt{4m-3})/2$ and this bound is
tight~\cite{Schuurman:Vredeveld:2007}. For lex-jump optimal
schedules, Czumaj and V\"ocking~\cite{Voecking:2007} showed that the performance guarantee is
$\Theta \big( \min \big\{ \frac{\log m}{\log \log m}, \log {\smax} \big\} \big)$.
For the problem on restricted related machines, Rutten et
al.~\cite{Rutten:etal:2012} showed that the performance guarantee of locally optimal schedules with
respect to the jump neighborhood is $(1+\sqrt{1+4(m-1) \smax})/2$ and
that this bound is tight up to a constant factor. Moreover, they showed
that the performance guarantee of lex-jump optimal schedules is
$\Theta \big( \frac{\log S}{\log \log S} \big)$, where $S = \sum_{i=1}^m s_i$.
When all speeds are equal, Awerbuch et al.~\cite{Awerbuch+etal:2006}
showed that the performance guarantee for lex-jump optimal schedules is
$\Theta \big( \frac{\log m}{\log \log m} \big)$.

Up to now, smoothed analysis has been mainly applied to running time
analysis (see, e.g.,~\cite{Spielman+Teng:CACM2009} for a survey). The first
exception is the paper by Becchetti et
al.~\cite{Becchetti+etal:Smoothed:MOR} who introduced the concept of
smoothed competitive analysis, which is equivalent to smoothed performance
guarantees for online algorithms. Sch\"afer and
Sivadasan~\cite{Schaefer+Sivadasan:2005} performed a smoothed
competitive analysis for metrical task systems. Englert et al.~\cite{Englert+etal:SODA2007} considered the $2$-Opt algorithm for the traveling salesman problem and determined, among others, the
smoothed performance guarantee of local optima of the $2$-Opt algorithm.
Hoefer and Souza~\cite{Hoefer+Souza:2010} presented one of the first
average case analyses for the price of anarchy.

The remainder of this article is organized as follows.
In Section~\ref{sec:relatedmachines}, we
provide asymptotically matching upper and lower bounds on the smoothed performance guarantees of
jump optimal,  lex-jump optimal, and list schedules in case of unrestricted related machines. In Section~\ref{sec:restricted}, we show that smoothing does not help for the setting of restricted machines.

\section{Unrestricted Related Machines}
\label{sec:relatedmachines}

\subsection{Jump Optimal Schedules}
\label{subsec:jump}

We show that the smoothed performance guarantee grows
linearly with the smoothing parameter~$\phi$ and is independent
of the number of jobs and machines. In particular, it is constant if the
smoothing parameter is constant. In proving our results, we make use of the 
following proposition which follows from Cho and Sahni~\cite{Cho+Sahni:1980}.

\begin{prop}
For any scheduling instance~$I$ with~$m$ unrestricted related machines and~$n$ jobs
\begin{equation}
\label{ChoSahni}
  \max_{\sched\in\JUMP{I}} \frac{\csched[I, \sched]}{\copt(I)}
  \leq \frac{1 + \sqrt{4 \min \{m, n\}-3}}{2} \leq \frac{1}{2} + \sqrt{n} \DOT
\end{equation}
\label{prop:ChoSahni}
\end{prop}

\begin{theorem}
\label{thm:JumpRelatedUpper}
For any $\phi$-smooth instance~$\mathcal{I}$ with unrestricted and related machines,
\[
    \El[I \sim \mathcal{I}]{\max_{\sched\in\JUMP{I}}\frac{\csched[I, \sched]}{\copt(I)}} < 5.1 \phi + 2.5 = O(\phi) \DOT
\]
\end{theorem}

\begin{proof}
First note that if~$m > n$, then there exist an optimal schedule and a worst
jump-optimal schedule that do not schedule any job on any of the
slowest~$m-n$ machines. We ignore these slowest~$m-n$ machines, and therefore
we assume that $m \leq n$.
We will prove an upper bound on the performance guarantee of jump optimal schedules
that decreases when the sum of processing requirements $Q=\sum_{j \in J} p_j$ increases and that is valid for every instance. Then, we will argue that for $\phi$-smooth instances~$Q$ is usually not too small, which
yields the theorem.

Let~$\sched$ denote an arbitrary jump optimal schedule for some arbitrary processing requirements~$p_j\in[0,1]$. Let~$i$ be an arbitrary machine, let machine~$\critmach$ be a critical machine in schedule~$\sched$, and let~$j$ be a job assigned to machine~$\critmach$ by schedule~\sched. By jump optimality of~\sched it follows that
\[
  \csched
  = \load{\critmach}
  \leq \load{i} + p_j / s_i
  \leq \load{i} + \pmax / s_i \COMMA
\]
where~$\pmax$ denotes the processing requirement of the largest job. The previous inequality yields that $s_i \cdot \csched \leq s_i \cdot \load{i} + \pmax$ for
all machines~$i \in M$.
Summing over all machines from $M \setminus \SET{ \critmach }$ and adding
$s_{\critmach} \cdot \load{\critmach}$ to both sides of the inequality, we find that
\[
  \sum_{i \in M} s_i \cdot \csched
  \leq \sum_{i \in M \setminus \{\critmach\}} \pmax + \sum_{i \in M} s_i \cdot \load{i} 
  \leq (n-1) \cdot \pmax + \sum_{i \in M} s_i \cdot \load{i} 
\]
since $L_{\critmach}= \csched$. Noting that $\sum_{i \in M} s_i \cdot \load{i} =
\sum_{j \in J} p_j = Q$ yields the following upper bound on the makespan of any jump optimal schedule~$\sched$:
\[
\csched \leq \frac{Q}{\sum_{i \in M} s_i} +  \frac{n-1}{\sum_{i \in M} s_i} \COMMA
\]
where the last inequality follows since~$\pmax \leq 1$. Using the well-known bound $\copt \geq Q/\sum_{i \in M} s_i$ we obtain
\[
  \csched  \leq \frac{Q}{\sum_{i \in M} s_i} +  \frac{n-1}{\sum_{i \in M} s_i}
  \leq \Big(1+\frac{n-1}{Q}\Big) \cdot \copt \DOT
\]
Hence,
\begin{equation}
\label{UB_jump_optimality}
\max_{\sched\in\JUMP{I}} \frac{\csched[I, \sched]}{\copt(I)} \le 1+\frac{n-1}{Q} \DOT
\end{equation}
The performance guarantee of any jump optimal schedule can only be bad if~$Q$ is small. Since the instance is $\phi$-smooth, the processing requirements are random
variables in~$[0,1]$ with bounded densities.
Let~$\mathcal{F}$ denote the failure event that $Q \leq (n-\sqrt{n\ln n})/(2\phi)$.
We define~$x_j$ to be independent random variables drawn uniformly from~$[0,1/\phi]$ for all~$j \in J$. Then,
$\Prob[p_j \geq a] \geq \Prob[x_j \geq a]$ for any~$a \in [0,1]$. 
Let $X = \sum_{j\in J} x_j$. Then, for any~$a\in [0,n]$,
it follows that $\Prob[Q \geq a] \geq  \Prob[X \geq a]$. Hence,
\begin{align}
  \Probl{\mathcal{F}}
  &= \Probl{Q \le \frac{n-\sqrt{n\ln n}}{2\phi}} 
  \le \Probl{X \le \frac{n-\sqrt{n\ln n}}{2\phi}} \notag \\
  &= \Probl{  \El{X} - X \geq   \frac{\sqrt{n \ln n}}{2 \phi} }
  \le \e^{- (\ln n) /2} = \frac{1}{\sqrt{n}} \COMMA
\label{eqn:probFailure}
\end{align}
where the last inequality follows from Hoeffding's bound~\cite{Hoeffding:1963} (see also Theorem~\ref{thm:app:hoeffding} in the appendix). Consider the random variable
\[
  Z = \begin{cases}
    \frac{1}{2} + \sqrt{n} & \text{if event} \ \mathcal{F} \ \text{occurs}  \COMMA \cr
    1 + \frac{n-1}{Q} & \text{otherwise}  \COMMA 
    
  \end{cases}
\]
and let $Y = \max_{\sched\in\JUMP{I}} \frac{\csched[I, \sched]}{\copt(I)}$. Due to Inequalities~\eqref{ChoSahni} and~\eqref{UB_jump_optimality} we have $Y \leq Z$. We denote by~$\overline{\mathcal{F}}$ the complement of~$\mathcal{F}$ and obtain
\begin{align*}
  \El[I \sim \mathcal{I}]{Y}
  &\leq \El[I \sim \mathcal{I}]{Z}
  \leq \El[I \sim \mathcal{I}]{\left. Z \right| \overline{\mathcal{F}} }   + \El[I \sim \mathcal{I}]{\left. Z \right|  \mathcal{F}} \cdot \Probl[I \sim \mathcal{I}]{\mathcal{F}} \\
  &\stackrel{\eqref{eqn:probFailure}}{\leq} \left( 1+\frac{2\phi(n-1)}{n-\sqrt{n\ln n}} \right) + \frac{1/2+\sqrt{n}}{\sqrt{n}}\\
  &<2.5+ \frac{2 \phi}{1 - \sqrt{\ln (n) / n }} < 2.5 + 5.1 \phi \DOT
\end{align*} 
For the third inequality, we used $Q > (n - \sqrt{n \ln n})/(2\phi)$ if event~$\mathcal{F}$ does not hold. The last inequality holds since 
\[ \max_{n \in \mathbb{Z}^+} \frac{2}{1- \sqrt{ \ln (n) /n}} < 5.1, \]
where the maximum is attained for $n=3$.
\end{proof}

\begin{cor}
Consider an instance of scheduling with unrestricted and related machines in which the processing requirement of every job is chosen independently and uniformly at random from $[0,1]$. The expected performance guarantee of the worst jump optimal schedule is~$O(1)$.
\end{cor}

Next, we show that the upper bound on the smoothed performance guarantee provided in Theorem \ref{thm:JumpRelatedUpper} is tight up to constant factor when~$\phi\geq 2$.

\begin{theorem}
\label{thm:JumpRelatedLower}
There is a class of $\phi$-smooth instances~$\mathcal{I}$ with unrestricted and related machines such that
\[
    \El[I \sim \mathcal{I}]{\max_{\sched\in\JUMP{I}}\frac{\csched[I, \sched]}{\copt(I)}} = \Omega(\phi) \DOT
\]
\end{theorem}

\begin{proof}
For any~$\phi > 2$ we construct a $\phi$-smooth instance~$\mathcal{I}$ with $n=\lceil 4\phi^2+1\rceil$ and~$m=n$ machines.
Let
\[
  s_1 = \frac{n-1}{4\phi} \geq \phi > 2
  \quad \text{and} \quad
  s_2 = \ldots = s_n =1 \DOT
\]
We assume that the processing requirement~$p_1$ is chosen uniformly from the interval $[1-1/\phi,1]$
while the processing requirements of all other jobs are chosen uniformly from the interval~$[0,1/\phi]$.
In an optimal schedule, job~$1$ is scheduled on machine~$1$, and all other
machines process exactly one job (see Figure~\ref{fig:jump-unrestricted-related-opt}). Hence,
\[
 \copt
  = \max \left\{ \frac{p_1}{s_1}, p_2, \ldots , p_n \right\}
  \leq \max \left\{ \frac{1}{s_1}, \frac{1}{\phi} \right\}
  =  \frac{1}{\phi} \DOT
\]
We show that with high probability there exists a jump optimal schedule~$\sched$
with $\csched > 1-1/\phi$. In order to find such a schedule~$\sched$, 
we first schedule job~$1$ on machine~$2$.
Then, we consider the remaining jobs one after another and schedule unassigned jobs to machine~$1$ until either
$\load{1} \in \big[ \load{2} - 1/(\phi s_1), \load{2} \big)$ or all jobs
are scheduled. Any job that remains unscheduled is then exclusively assigned to one
empty machine. Let~\event denote the event that $Q_2:=\sum_{j =2}^n p_j \geq s_1$. Note that $\El{Q_2} = (n-1)/(2\phi) = 2s_1$. We will see that event~\event holds with high probability with respect to~$\phi$.

\begin{artclfig}\newcommand{\height}{11em}
  \begin{minipage}[t]{0.55\textwidth}\begin{center}
    \includegraphics[height=\height]{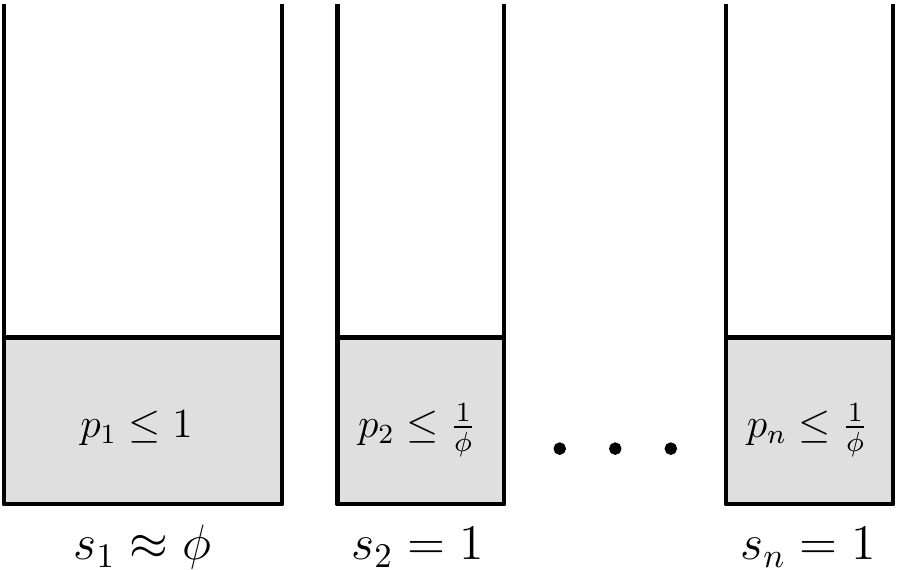}
    \caption{Optimal schedule}
    \label{fig:jump-unrestricted-related-opt}
  \end{center}\end{minipage}\hfill
  \begin{minipage}[t]{0.4\textwidth}\begin{center}
    \includegraphics[height=\height]{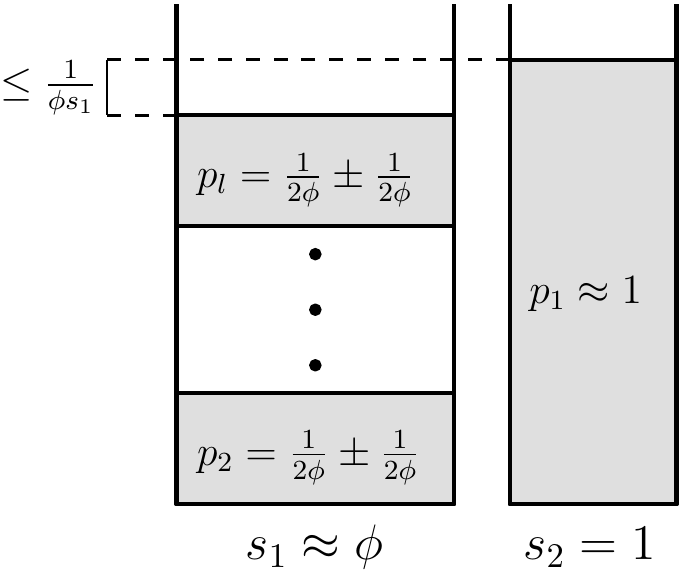}
    \caption{Machines~$1$ and~$2$ of schedule~\sched if event~\event occurs}
    \label{fig:jump-unrestricted-related-jump-opt}
  \end{center}\end{minipage}
\end{artclfig}

Consider the case that event~\event occurs. Then, schedule~$\sched$ is such that $\load{1} \in \big[ \load{2} - 1/(\phi s_1), \load{2} \big)$ since $Q_2/s_1\geq 1 \geq p_1= \load{2}$ and $p_j \leq 1/\phi$ for
all jobs $j = 2, \ldots, n$ (see Figure~\ref{fig:jump-unrestricted-related-jump-opt}). Now, we argue that schedule~$\sched$ is jump optimal.
First observe that machine~$2$ defines the makespan since $L_2 > \max \SET{ L_1, p_2/1, \ldots, p_n/1 }$. Job~$1$, which is the only job assigned to
that machine, cannot jump to a machine~$i>2$ because these have the same speed as machine~$2$. Furthermore,
it cannot jump to machine~$1$ because
\[
  \load{1} + \frac{p_1}{s_1}
  \geq \load{2} - \frac{1}{\phi s_1} + \frac{1-1/\phi}{s_1}
  = \load{2}  + \frac{1-2/\phi}{s_1}
  > \load{2}
\]
as~$\phi > 2$. Hence, $\sched$ is a jump optimal schedule with 
\begin{equation}
\label{eq:lb-jump-unrestricted}
  \frac{\csched}{\copt}
  > \frac{1-1/ \phi}{1/\phi}
  = \phi -1 \DOT
\end{equation}
It remains to determine the probability of event~$\event$. Recalling $\El{Q_2} = 2s_1$, $s_1 = (n-1)/(4\phi)$, and $n \geq 4\phi^2 + 1$, this can be bounded with Hoeffding's bound~\cite{Hoeffding:1963} (see also Theorem~\ref{thm:app:hoeffding}) as follows:
\begin{align*}
  \Probl{\overline{\event}}
  &= \Probl{Q_2 < s_1}
  = \Probl{\El{Q_2} - Q_2 > s_1}
  \leq \exp \left(\frac{-2s_1^2}{(n-1)/\phi^2} \right)\\
  &= \exp \left( {\frac{-2\left( \frac{n-1}{4 \phi} \right)^2}{(n-1)/\phi^2}} \right)
  = \exp \left( -\frac{n-1}{8} \right)
  \leq \exp \left( -\frac{\phi^2}{2} \right) \DOT
\end{align*}
Let $X = \max_{\sched\in\JUMP{I}}\frac{\csched[I, \sched]}{\copt(I)}$. Applying Inequality~\eqref{eq:lb-jump-unrestricted} the smoothed performance guarantee can be bounded from below as follows: 
\begin{align*}
  \El[I \sim \mathcal{I}]{X}
  &\ge \El[I \sim \mathcal{I}]{\left. X \right| \event} \cdot \Probl[I \sim \mathcal{I}]{\event}
  \geq \left( \phi - 1 \right) \cdot \left(1 - \exp \left( -\frac{\phi^2}{2} \right) \right) \\
  &= (\phi -1) - (\phi -1) \cdot \exp \left( -\frac{\phi^2}{2} \right)
  > \phi - 1.14
  = \Omega(\phi) \COMMA
\end{align*}
where the last inequality follows because $(\phi -1) \cdot \exp(-\phi^2/2) < 0.14$ for~$\phi > 2$.
\end{proof}

\subsection{Upper Bounds for List Schedules and Lex-jump Optimal Schedules}
\label{subsec:ub-list-lex-jump}

Although the worst case performance bound on unrestricted related machines for list scheduling is slightly worse than the one for lex-jump scheduling, we show that the smoothed performance guarantee of both schedules is~$O(\log \phi)$. In the next subsection, we show that this bound is asymptotically tight.

\begin{theorem}
\label{thm.mainI.list.lex}
Let~$\alpha$ be an arbitrary positive real. For~$\phi \geq 2$ and any $\phi$-smooth instance~$\mathcal{I}$ with unrestricted and related machines
\[ \Probl[I \sim \mathcal{I}]{\max_{\sched \in \LEX{I} \cup \LIST{I}}
\frac{\csched[I, \sched]}{\copt(I)} \geq \alpha} \leq \left( \frac{32\phi}{2^{\alpha/6}} \right)^{n/2} \]
and
\[ \El[I \sim \mathcal{I}]{\max_{\sched \in \LEX{I} \cup \LIST{I}}
\frac{\csched[I, \sched]}{\copt(I)}} \leq 18 \log_2 \phi + 30 = O(\log \phi) \DOT \]
\end{theorem}

Note that the assumption~$\phi \geq 2$ in
Theorem~\ref{thm.mainI.list.lex} is no real restriction as for $\phi
\in [1, 2)$ any $\phi$-smooth instance is a $2$-smooth instance. Hence,
for these values we can apply all bounds from
Theorem~\ref{thm.mainI.list.lex} when substituting~$\phi$ by~$2$. In
particular, the expected value is a constant.

In the remainder of this section, we will use the following notation (see also Appendix~\ref{sec:appendix-table}).
Let $\jobset[i,j](\sched)$ denote the set
of all jobs that are scheduled on machine~$i$ and have index at most~$j$, i.e.,
$\jobset[i,j](\sched) = \jobsetsched{i} \cap \SET{ 1,\ldots, j }$. If~$\sched$ is clear
from the context, then we just write~$\jobset[i,j]$.
We start with observing an essential property that both lex-jump optimal
schedules and list schedules have in common.

\begin{definition}
We call a schedule~$\sched$ on machines $1, \ldots, m$ with speeds $s_1,
\ldots, s_m$ a \emph{near list schedule}, if we can index the jobs in
such a way that
\begin{equation}
\label{eq.key.property}
\load{i'} + \frac{p_j}{s_{i'}} \geq \load{i} -
\sum \limits_{\ell \in \jobset[i,j-1](\sched)} \frac{p_\ell}{s_i}
\end{equation}
for all machines~$i' \neq i$ and all jobs~$j \in \jobs{i}$. With $\NL{I}$ we denote the set of all near list schedules for instance~$I$.
\end{definition}

Inequality~\eqref{eq.key.property} can be interpreted as follows. Assume that the jobs are already indexed correctly and imagine that on each machine the jobs form a stack, ordered from top to bottom ascendingly according to their index. Now, consider an arbitrary job~$j$ on machine~$i$ (see Figure~\ref{fig:near-list-schedule1}). Inequality~\eqref{eq.key.property} states that the completion time of job~$j$ after removing all jobs above~$j$ is minimized on machine~$i$ in case only job~$j$ is allowed to move (see Figure~\ref{fig:near-list-schedule2}).

\begin{artclfig}\newcommand{\height}{10em}
  \subfloat[Jobs on machine~$i$, including job~$j$, visualized as a stack]{
    \includegraphics[height=\height, page=1]{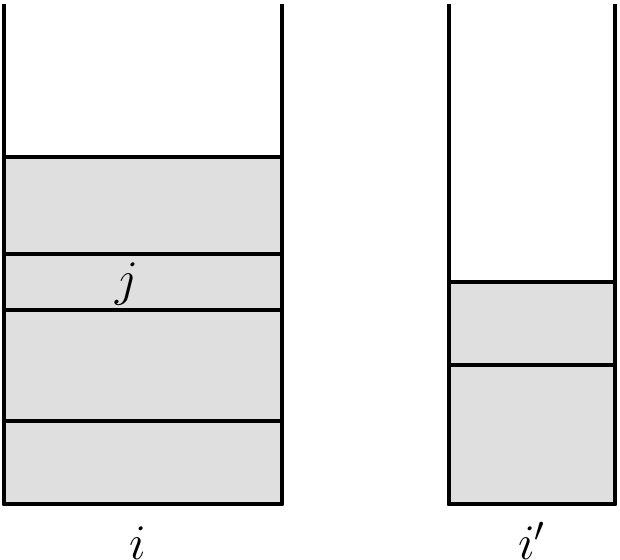}
    \label{fig:near-list-schedule1}
  }
  \hspace{5ex}
  \subfloat[Job~$j$ does not benefit from jumping to machine~$i'$]{
    \includegraphics[height=\height, page=2]{NearListSchedule.pdf}
    \label{fig:near-list-schedule2}
  }
  \caption{Interpretation of Inequality~\eqref{eq.key.property}}
\end{artclfig}

\begin{lemma}
\label{lemma.NL}
For any instance~$I$ the relation $\LEX{I} \cup \LIST{I} \subseteq
\NL{I}$ holds.
\end{lemma}

Note that in general neither $\LEX{I} \subseteq \LIST{I}$ nor $\LIST{I} \subseteq \LEX{I}$ holds (see Figure~\ref{fig:strict-supset}). 
 Moreover, there also exist near list schedules that are neither in $\LEX{I}$ nor in $\LIST{I}$ (see Figure~\ref{fig:strict-supset generalization}), i.e., near list schedules are a non-trivial generalization of both lex-jump optimal schedules and list schedules.

\begin{artclfig}
\newcommand{\height}{10em}
  \subfloat[A list schedule which is not lex-jump optimal]{
    \includegraphics[height=\height]{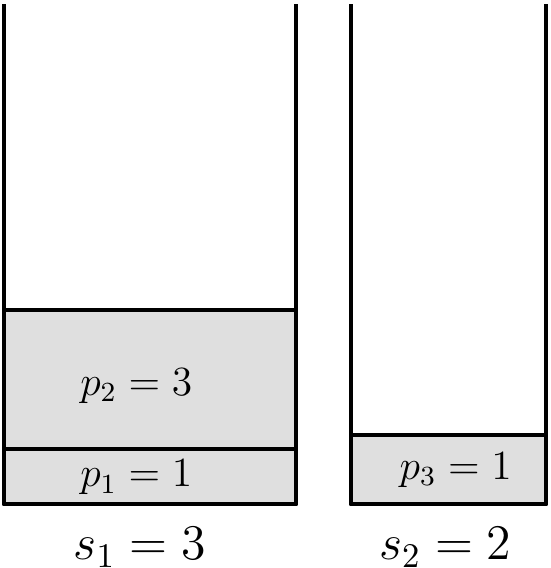}
  }
  \hspace{5ex}
  \subfloat[A lex-jump optimal schedule which is no list schedule]{
    \includegraphics[height=\height]{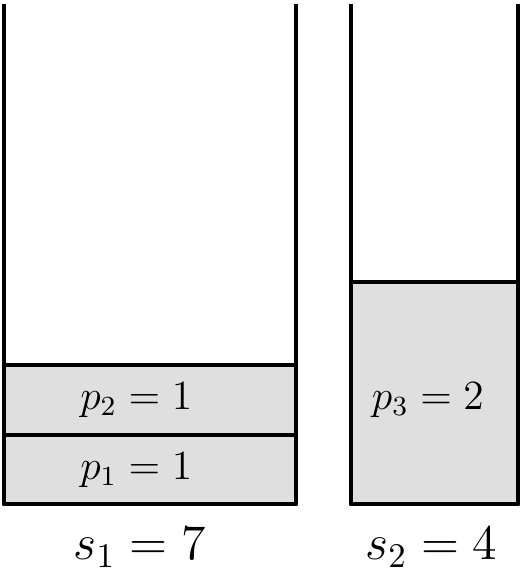}
  }
  \hspace{5ex}
  \subfloat[A near list schedule which is neither lex-jump optimal nor a list schedule]{
    \includegraphics[height=\height]{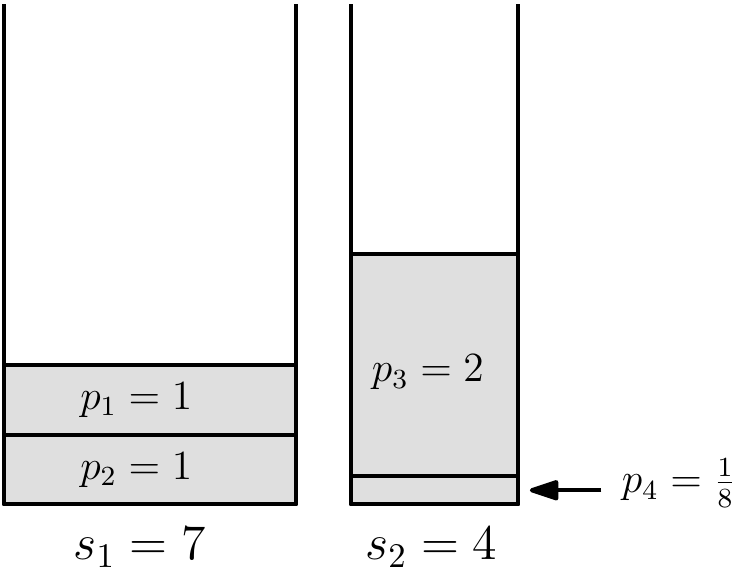}
    \label{fig:strict-supset generalization}
  }
  \caption{Relationship between $\LEX{I}$, $\LIST{I}$, and $\NL{I}$}
  \label{fig:strict-supset}
\end{artclfig}

\begin{proof}[\ProofText{Lemma~\ref{lemma.NL}}]
For any schedule~$\sched \in \LEX{I}$, we can index the jobs
arbitrarily and, by definition, even the stronger inequality $\load{i'}
+ p_j/s_{i'} \geq \load{i}$ holds. For~$\sched \in \LIST{I}$ we can
index the jobs in reverse order in which they appear in the list that
was used for list scheduling.
Consider an arbitrary job $j \in \jobs{i}$ and a machine $i'
\neq i$. Let $\load{i}', \load{i'}'$ and $\load{i}, \load{i'}$ denote
the loads of machines~$i$ and~$i'$ before assigning job~$j$ to
machine~$i$ and the loads of~$i$ and~$i'$ in the final schedule,
respectively. Then, $\load{i}' + p_j/s_i \leq \load{i'}' + p_j/s_{i'}$
as~$j$ is assigned to machine~$i$ according to list scheduling. Since
$\load{i} = \load{i}' + \sum_{\ell \in \jobset[i,j]} p_\ell/s_i$ and $\load{i'} \geq
\load{i'}'$, this implies $\load{i'} + p_j/s_{i'} \geq \load{i'}' + p_j / s_{i'} \geq \load{i}' + p_j / s_i = \load{i} - \sum_{\ell \in \jobset[i,j-1]} p_\ell/s_i$.
\end{proof}

In the remainder, we fix an instance~$I$ and consider an arbitrary
schedule~$\sched \in \NL{I}$ with appropriate indices of the jobs
such that Inequality~\eqref{eq.key.property} holds.
To prove Theorem~\ref{thm.mainI.list.lex}, we show
that in case the ratio of $\csched[I, \sched]$ over $\copt(I)$ is large,
then instance~$I$ needs to have many very small jobs,
see~Corollary~\ref{corol.many.small.jobs}.
This holds even when  the instance $I$ is deterministically picked by some adversary.
This observation allows us to prove the main theorem of this subsection by showing that for any $\phi$-smooth instance, there are only ``few'' small jobs in expectation. The latter implies that a large ratio only happens with (exponentially) small probability.

In our proofs, we adopt some of the notation also used by Czumaj and
V\"{o}cking~\cite{Voecking:2007} (see also Appendix~\ref{sec:appendix-table}). 
Given a schedule~$\sched$, we set $c = \FLOOR{ \csched/\copt } - 1$. Recall that the
machines are ordered such that $s_1 \geq \hdots \geq s_m$. For any
integer~$k \leq c$ let $H_k = \SET{ 1, \ldots, i_k }$ where $i_k = \MAX{ i \in M \WHERE \load{i'}
\geq k \cdot \copt \, \forall \, i' \leq i }$. Note that $i_k = m$ for all~$k
\leq 0$ and hence~$H_k = M$ for such~$k$ (see Figure~\ref{fig:machine-classification}). Further, define $\H{k} = H_k
\setminus H_{k+1}$ for all $k \in \SET{ 0, \ldots, c-1 }$ and $\H{c} = H_c$.
Note that this classification always refers to schedule~$\sched$ even if
additionally other schedules are considered. Some properties follow
straightforwardly.

\begin{artclfig}\newcommand{\height}{14em}
\includegraphics[height=\height]{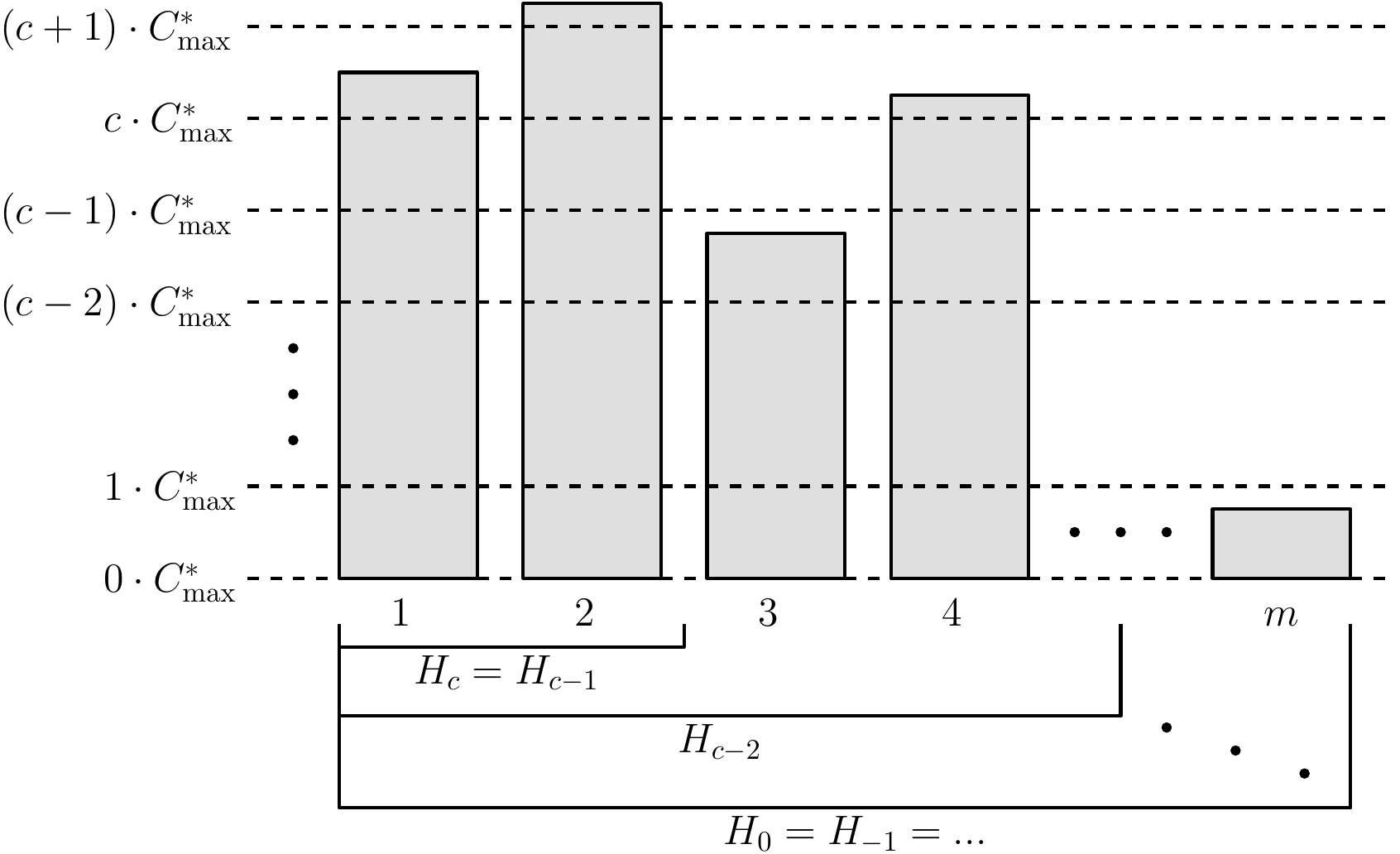}
\caption{Machine classification by Czumaj and V{\"o}cking}
\label{fig:machine-classification}
\end{artclfig}

\begin{property}
\label{prop:minimum-load}
For each machine~$i \in H_k$, $\load{i} \geq k \cdot \copt$.
\end{property}

\begin{property}
Machine~$i_k +1$, if it exists, is the first machine in $M \setminus H_k$, i.e., the machine with the least index,
and, hence, a fastest machine in $M \setminus H_k$.
\end{property}

\begin{property}
\label{prop:maximum-load}
$\load{i_k +1} < k \cdot \copt$ for all $k \in \SET{ 1, \ldots, c }$, and $\load{1} < (c+2) \cdot \copt$.
\end{property}

As mentioned, we need to show that there are many small jobs. To do so,
we will show that the the speeds of the machines in low classes, i.e., $R_0$ and
$R_1$, are exponentially small with respect to the machines in the
highest class $R_c$ (Lemma~\ref{lemma.no.neighboring.classes.empty})
and that the machines in low classes
need to process high volume (Lemma~\ref{lemma.job.redistribution}).
We start by showing that the highest class is nonempty.

\begin{lemma}
\label{lem:nonemptyHc}
Machine~$1$ is in class~$\H{c}$.
\end{lemma}

\begin{proof}
Let~$i$ be a critical machine. If~$i = 1$, then we obtain $L_1/\copt = \csched/\copt > c$. Otherwise we apply Inequality~\eqref{eq.key.property} for the job~$j = \MIN{ \ell \in
\jobset[i]}$ with the smallest index on machine~$i$ and for machine~$1$. This yields $\load{1} + p_j/s_1 \geq \load{i}$. Hence,
$\load{1}/\copt \geq \load{i}/\copt - (p_j/s_1)/\copt \geq \csched/\copt - 1 \geq c$,
where the second inequality is due to the fact that any job can
contribute at most~$\copt$ to the makespan of a fastest machine.
\end{proof}

Let $t$ and $k$ be integers satisfying $0 \leq t \leq k \leq c$. Several times we will consider the first many jobs 
on some machine~$i \in H_k$ which contribute at least~$t \cdot
\copt$ to the load of machine~$i$. We denote the set of those jobs
by~$\jobset[i,\geq t]$. Formally, 
$$
\jobset[i,\geq t] = \jobset[i,j_i^t] \text{ for } j_i^t
= \min \big\{ j \WHERE \sum_{\ell \in \jobset[i,j]} p_\ell/s_i \geq t \cdot \copt
\big\}.
$$
Using this notation, Lemma \ref{lemma:main-prop-near-list-schedule} and Corollary \ref{corol.optimal.assignment} restrict the machines on which a job in $\jobset[i,\geq t]$ can be scheduled in an optimal schedule.

\begin{lemma}
\label{lemma:main-prop-near-list-schedule}
Let~$k_1 > k_2$ and~$t \leq k_1$ be positive integers, let $i_1 \in H_{k_1}$ and $i_2 \in M \setminus H_{k_2}$ be  machines in~$H_{k_1}$ and not in~$H_{k_2}$, respectively, and let $j \in \jobset[i_1,\geq t]$ be a job on machine~$i_1$. Then, the load job~$j$ would contribute to machine~$i_2$ is bounded from below by $p_j/s_{i_2} > (k_1-k_2-t) \cdot \copt$.
\end{lemma}

\begin{proof}
We apply Inequality~\eqref{eq.key.property} for machine~$i_1$, for the first machine~$i'_2$ that does not belong to~$H_{k_2}$, and for job~$j$ to obtain
\[
  L_{i'_2} + \frac{p_j}{s_{i'_2}} \geq L_{i_1} - \sum_{\ell \in \jobset[i_1,j-1]} \frac{p_\ell}{s_{i_1}} \COMMA
\]
which implies
\[
  \frac{p_j}{s_{i'_2}} \geq L_{i_1} - L_{i'_2} - \sum_{\ell \in \jobset[i_1,j-1]} \frac{p_\ell}{s_{i_1}} \DOT
\]
By the choice of the machines~$i_1$ and~$i'_2$ and Properties~\ref{prop:minimum-load} and~\ref{prop:maximum-load} we obtain $L_{i_1} \geq k_1 \cdot \copt$ and $L_{i'_2} < k_2 \cdot \copt$. Furthermore, $j \in \jobset[i_1,\geq t]$ yields $\sum_{\ell \in \jobset[i_1,j-1]} p_\ell/s_{i_1} < t \cdot \copt$. Hence, $p_j/s_{i'_2} > (k_1-k_2-t) \cdot \copt$. The claim follows since $s_{i'_2} \geq s_{i_2}$.
\end{proof}

\begin{cor}
\label{corol.optimal.assignment}
Let~$i \in H_k$ be an arbitrary machine and let~$t \in \SET{ 1, \ldots, k }$ be an integer. Then, in any optimal schedule any job $j \in \jobset[i,\geq t]$ is assigned to machines from $H_{k-t-1}$.
\end{cor}

\begin{proof}
Assume, for contradiction, that there is a job $j \in \jobset[i,\geq t]$ which is assigned to a machine~$i' \in M \setminus H_{k-t-1}$ by an optimal schedule. By Lemma~\ref{lemma:main-prop-near-list-schedule} this job causes a load of more than $(k - (k-t-1) - t) \cdot \copt = \copt$ on this machine contradicting the assumption that the considered schedule is optimal.
\end{proof}

Czumaj and V\"{o}cking~\cite{Voecking:2007} showed that in a lex-jump optimal
schedule the speeds of any two machines which are at least two classes apart
differ by a factor of at least~$2$. Aspnes et~al.~\cite{DBLP:journals/jacm/AspnesAFPW97} showed a similar property. In general, near list schedules have a
slightly weaker property. 

\begin{lemma}
\label{lemma.decreasing.speeds}
Let~$k \in \SET{ 5, \ldots, c }$ and assume $H_k \neq \emptyset$. The speed of any machine in
class~$H_k$ is at least twice the speed of any machine in~$M \setminus H_{k-4}$.
\end{lemma}

\begin{proof}
We may assume that $M \setminus H_{k-4} \neq \emptyset$, since otherwise all machines have a load larger than $\copt$ as $H_k \neq \emptyset$. Let $i_0 \in H_k$ and $i_2 \in M \setminus H_{k-4}$ be arbitrary machines and consider the jobs from $\bigcup_{i \in H_k} \jobset[i,\geq 2]$. If we would assign only these jobs to machines in~$H_k$, then
there would be a machine with load at least $2 \cdot \copt$. Consequently, in an optimal schedule at least one job in $\bigcup_{i' \in H_k} \jobset[i',\geq 2]$ is assigned to some machine $i^* \in M \setminus H_k$, say job $j \in J_{i_1,\ge 2}$. Since job~$j$ contributes at most~$\copt$ to the load of machine~$i^*$ in this optimal schedule, this implies $p_j/s_{i^*} \leq \copt$ and, hence,
\begin{equation}
\label{eq:fast-processing}
\frac{p_j}{s_{i_0}} \leq \copt
\end{equation}
as $s_{i_0} \geq s_{i^*}$.
Due to Lemma~\ref{lemma:main-prop-near-list-schedule}, the load that would be contributed by job~$j$ on machine~$i_2$ is bounded by $p_j/s_{i_2} > (k - (k-4) - 2) \cdot \copt = 2 \cdot \copt$. Inequality~\eqref{eq:fast-processing} yields $s_{i_0} \geq 2 \cdot s_{i_2}$ as claimed in the lemma.
\end{proof}

We want to show that machines in low classes, i.e., machines in $\H{0} \cup \H{1}$, have exponentially small speeds
(with respect to~$c$) compared to the speeds of the machines in a high class, i.e., those in~$\H{c}$.
Lemma~\ref{lemma.decreasing.speeds} already implies that the machine speeds
would double every five classes if no class~$\H{k}$ was empty. Although some classes~$\H{k}$ can be empty,
we show that not too many of these machine classes are empty. This is done in Lemma~\ref{lemma.no.neighboring.classes.empty}
which follows from the next lemma.

The machines $i \in H_k$, $k \geq 2$, are overloaded compared to an optimal schedule,
even if we just consider the first few jobs $j \in \jobset[i,\geq t]$ on them (where $t \geq 2$). On the other hand, in Corollary~\ref{corol.optimal.assignment} we showed that in any optimal schedule these jobs are not assigned to machines in much lower classes, i.e., to machines from $M \setminus H_{k - t- 1}$. Consequently, in any optimal schedule the machines in $H_{k-t-1} \setminus H_k$ consume the current overload of~$H_k$.

\begin{lemma}
\label{lemma.job.redistribution}
Let~$t \leq k$ be positive integers. In any optimal schedule the total processing requirement on all machines in~$H_{k-t-1} \setminus H_k$ is at least
\[  \sum \limits_{k'=k}^c \sum \limits_{i \in \H{k'}} (t+k'-k-1) \cdot s_i \cdot \copt  \DOT \]
\end{lemma}

Note that Lemma~\ref{lemma.job.redistribution} also holds for the case $t = k$ where $H_{k-t-1} = H_{-1} = M$.

\begin{proof}
Applying Corollary~\ref{corol.optimal.assignment} with $t'(k') = t+(k'-k)$ for arbitrary integers~$k' \in \SET{ k, \ldots, c }$ yields that in any optimal schedule~$\schedopt$ all jobs in
$\bigcup_{k'=k}^c \bigcup_{i \in \H{k'}} \jobset[i,\geq t'(k')]$ are assigned to machines in~$H_{k-t-1}$ as
$k' - t'(k') - 1 = k-t-1$ for any index~$k'$. Furthermore, in~$\schedopt$
the processing requirement on any machine~$i \in H_k$ is at most $s_i
\cdot \copt$, i.e., the machines in $H_{k-t-1} \setminus H_k$ must consume the remainder. Hence, these machines must process jobs with total processing requirement at least
\[
  \sum \limits_{k'=k}^c \sum \limits_{i \in \H{k'}} \sum_{\ell \in \jobset[i,\geq t'(k')]}
p_{\ell} - \sum \limits_{k'=k}^c \sum \limits_{i \in \H{k'}} s_i \cdot \copt
  \geq \sum \limits_{k'=k}^c \sum \limits_{i \in \H{k'}} (t'(k')-1) \cdot s_i \cdot \copt \DOT
\]
This yields the claimed bound as $t'(k') - 1 = t+k'-k-1$.
\end{proof}

Although some machine classes~$\H{k}$ might be empty, we are able to show that this cannot be the case for two consecutive classes.

\begin{lemma}
\label{lemma.no.neighboring.classes.empty}
$H_{k-2} \setminus H_k \neq \emptyset$ for any $k \in \SET{ 1, \ldots, c - 1 }$.
\end{lemma}

\begin{proof}
Let~$i'$ be a slowest machine in~$H_k$. In any optimal schedule~$\schedopt$
the processing requirement on any machine~$i \in H_{k-2} \setminus H_k$
is at most $s_i \cdot \copt \leq s_{i'} \cdot \copt$.
Applying Lemma~\ref{lemma.job.redistribution} with~$t = 1$ implies
\[
  |H_{k-2} \setminus H_k| \cdot s_{i'} \cdot \copt
  \geq \sum \limits_{k'=k}^c \sum \limits_{i \in \H{k'}} (k'-k) \cdot s_i \cdot
\copt
  \geq \sum \limits_{k'=k}^c (k'-k) \cdot s_{i'} \cdot \copt \cdot |\H{k'}| \DOT
\]
It follows that 
\[
  |H_{k-2} \setminus H_k|
  \geq \sum \limits_{k'=k}^c (k'-k) \cdot |\H{k'}|
  \geq (c-k) \cdot |\H{c}|
  \geq 1
\]
since~$k < c$ and since $\H{c} \neq \emptyset$ due to Lemma~\ref{lem:nonemptyHc}.
\end{proof}

We can now show that machine speeds double every six classes. To be more formal:

\begin{lemma}
\label{lemma.exponentially.decreasing.speeds}
Let $0 \leq k_2 \leq k_1 \leq c$ be integers, let~$i_1$ be any machine of~$\H{k_1}$
and let~$i_2 \in \H{k_2}$. Then, $s_{i_1} \geq s_{i_2} \cdot 2^{\FLOOR{\Delta/6}}$ where $\Delta = k_1 - k_2$.
\end{lemma}

\begin{proof}
We prove the claim by induction. For $\Delta \in \SET{ 0, \ldots, 5 }$, the claim trivially holds as $s_{i_1} \geq s_{i_2}$. Assume that the claim holds up to some integer $\Delta^* \geq 5$. We show that it is also true for $\Delta = \Delta^* + 1 \geq 6$. Note that for such~$\Delta$ we have $k_1 \geq 6$. According to Lemma~\ref{lemma.no.neighboring.classes.empty} the class $H_{k_1-6} \setminus H_{k_1-4} \subseteq M \setminus H_{k_1 -4}$ contains at least one machine. Let~$i'$ be the fastest machine in $H_{k_1-6} \setminus H_{k_1-4}$. Then $s_{i'} \geq s_{i_2}$. Lemma~\ref{lemma.decreasing.speeds} and the induction hypothesis imply $s_{i_1} \geq 2 s_{i'}$ and $s_{i'} \geq s_{i_2} \cdot 2^{\FLOOR{(\Delta-6)/6}}$, respectively. Hence, $s_{i_1} \geq s_{i_2} \cdot 2^{\FLOOR{\Delta/6}}$.
\end{proof}

Since the machines in low classes are exponentially slower than the machines in high classes (with respect to~$c$) and as their aggregated total processing requirement in an optimal schedule is large (Lemma~\ref{lemma.job.redistribution}), it follows that many jobs have processing requirements exponentially small in~$c$. 

\begin{lemma}
\label{lemma.small.jobs}
Let $i \in M \setminus H_2$ be an arbitrary machine. Then each job~$j$ assigned to machine~$i$ by an optimal schedule has processing requirement at most $p_j \leq 2^{-c/6+2}$.
\end{lemma}

\begin{proof}
For $c \leq 12$ the claim is true since we rescale all processing requirements to be at most~$1$. Assume $c \geq 13$. Consider an optimal schedule~$\schedopt$ and
let~$j$ be a job processed on a machine~$i \in M \setminus H_2 = \H{1} \cup \H{0}$ according to~$\schedopt$. Note that $M \setminus H_2 \neq \emptyset$ due to Lemma~\ref{lemma.no.neighboring.classes.empty}. Then, $p_j/s_i \leq
\copt$, i.e.,
\begin{equation}
\label{eq:low-processing-requirement}
p_j \leq s_i \cdot \copt \DOT
\end{equation}
To bound $s_i \cdot \copt$, consider the job $j' = \min \{\ell \in \jobset[1](\sched) \}$ with the smallest index on machine~$1$ of schedule~$\sched$ and consider the first machine~$i' \in H_{c-3} \setminus H_{c-1} = \H{c-3} \cup \H{c-2}$ which exists due to Lemma~\ref{lemma.no.neighboring.classes.empty} and $c \geq 13$. Applying Inequality~\eqref{eq.key.property}, we obtain $\load{i'}(\sched) + p_{j'}/s_{i'} \geq \load{1}(\sched)$, i.e., $p_{j'} \geq s_{i'} \cdot (\load{1}(\sched) - \load{i'}(\sched))$. Since machine~$1$ belongs to~$H_c$ due to Lemma~\ref{lem:nonemptyHc} and since machine~$i'$ is the first machine that does not belong to~$H_{c-1}$, we have $\load{1}(\sched) \geq c \cdot \copt$ and $\load{i'}(\sched) < (c-1) \cdot \copt$, which implies $p_{j'} \geq s_{i'} \cdot \copt$. Lemma~\ref{lemma.exponentially.decreasing.speeds} yields $s_{i'} \geq s_i \cdot 2^{\FLOOR{(c-3-1)/6}}$. Applying Inequality~\ref{eq:low-processing-requirement} and $p_{j'} \leq 1$ according to our input model we obtain
\[
  p_j
  \leq s_i \cdot \copt
  \leq s_{i'} \cdot \copt \cdot 2^{-\FLOOR{(c-4)/6}}
  \leq p_{j'} \cdot 2^{-c/6+2}
  \leq 2^{-c/6+2} \DOT \qedhere
\]
\end{proof}

\begin{cor}
\label{corol.many.small.jobs}
The processing requirement of at least~$n/2$ jobs is at most $2^{-c/6+2}$.
\end{cor}

\begin{proof}
Lemma~\ref{lemma.job.redistribution} for $k = t = 2$ implies that the total
processing requirement of all jobs assigned to machines from $M \setminus H_2 = H_{-1} \setminus
H_2$ according to~$\schedopt$ is at least $\sum_{i \in H_2} s_i \cdot \copt$
which is an upper bound for the total processing requirement of all jobs
assigned to machines in~$H_2$ according to~$\schedopt$. Since all jobs assigned to machines from $M \setminus H_2$ by an optimal schedule have processing requirement at most $2^{-c/6+2}$ due to Lemma~\ref{lemma.small.jobs}, at
least half of the jobs have processing requirement at most~$2^{-c/6+2}$.
\end{proof}

Since having many so small jobs is unlikely when the processing
requirements have been smoothed, it follows that the smoothed
performance guarantee, which is between $c+1$ and $c+2$, cannot be too high, yielding Theorem~\ref{thm.mainI.list.lex}.

\begin{proof}[\ProofText{Theorem~\ref{thm.mainI.list.lex}}]
If $\csched/\copt \geq \alpha$, then at least~$n/2$ jobs have processing requirement at most $2^{-\alpha/6+3}$ due to Corollary~\ref{corol.many.small.jobs} and $c = \FLOOR{\csched/\copt} - 1 \geq \alpha - 2$. The probability that one specific job is that small is bounded by $\phi \cdot 2^{-\alpha/6+3} = 8\phi \cdot 2^{-\alpha/6}$ in the smoothed input model. Hence, the probability that the processing requirement of at least~$n/2$ jobs is at most $2^{-\alpha/6+3}$, is bounded from above by
\begin{align*}
  \sum_{k \geq \frac{n}{2}} &\binom{n}{k} \left( 8\phi \cdot 2^{-\alpha/6} \right)^k \cdot \left( 1 - 8\phi \cdot 2^{-\alpha/6} \right)^{n-k}
  \leq \sum_{k \geq \frac{n}{2}} \binom{n}{k} \left( 8\phi \cdot 2^{-\alpha/6} \right)^{n/2} \cr
  &\leq 2^n \cdot \left( 8\phi \cdot 2^{-\alpha/6} \right)^{n/2}
  = \left( 32\phi \cdot 2^{-\alpha/6} \right)^{n/2} \DOT
\end{align*}
Note that the first inequality holds if $8\phi \cdot 2^{-\alpha/6} < 1$. Otherwise, the bound is trivially true. This yields
\[
  \Probl[I \sim \mathcal{I}]{\max_{\sched \in \NL{I}}\frac{\csched[I, \sched]}{\copt(I)} \geq \alpha} 
\leq \left( \frac{32\phi}{2^{\alpha/6}} \right)^{n/2} \DOT
\]
As for~$n = 1$ any schedule~$\sched \in \NL{I}$ is optimal, we just consider the case~$n \geq 2$. For $k \geq 1$ let $\alpha_k = \alpha_k(\phi) = 6k \log_2 \phi + 30$, i.e., $2^{\alpha_k/6} = 32\phi^k$. If $\alpha \geq \alpha_k$, then we obtain
\begin{align*}
  \Probl[I \sim \mathcal{I}]{\max_{\sched \in \NL{I}}\frac{\csched[I, \sched]}{\copt(I)} \geq \alpha}
  &\leq \Probl[I \sim \mathcal{I}]{\max_{\sched \in \NL{I}}\frac{\csched[I, \sched]}{\copt(I)} \geq \alpha_k} \cr
  &\leq \left( \phi^{1-k} \right)^{n/2}
  \leq \phi^{1-k} \leq 2^{1-k}
\end{align*}
as $\phi \geq 2$. Since $\alpha_{k+1} - \alpha_k = 6 \log_2 \phi$ we obtain
\begin{align*}
  \El[I \sim \mathcal{I}]{\max_{\sched \in \NL{I}} \frac{\csched[I, \sched]}{\copt(I)}} 
  &= \int_0^\infty \Probl[I \sim \mathcal{I}]{\max_{\sched \in \NL{I}} \frac{\csched[I, \sched]}{\copt(I)} \geq \alpha} \d\alpha \cr
  &\leq \alpha_1 + \sum \limits_{k=1}^\infty \int_{\alpha_k}^{\alpha_{k+1}}
    \Probl[I \sim \mathcal{I}]{\max_{\sched \in \NL{I}} \frac{\csched[I, \sched]}{\copt(I)} \geq \alpha} \d\alpha \cr
  &\leq \alpha_1 + 6 \log_2 \phi \cdot \sum \limits_{k=1}^\infty 2^{1-k}
  = 18 \log_2 \phi + 30 \DOT \qedhere
\end{align*}
\end{proof}

\subsection{Lower Bounds for List Schedules and Lex-jump Optimal Schedules}
\label{subsec:lb-list-lex-jump}

In this subsection, we show that the upper bound,
given in Theorem~\ref{thm.mainI.list.lex},
on the smoothed performance guarantee on  lex-jump optimal as well as list schedules is tight up to a constant factor.
We provide a $\phi$-smooth instance such that the worst lex-jump optimal schedule as well as the worst schedule that can be obtained by list scheduling has a lower bound on the performance guarantee of $\Omega(\log \phi)$, for any realization of the processing times.

\begin{theorem}
\label{thm.mainII.list.lex}
There is a class of $\phi$-smooth instances~$\mathcal{I}$ with unrestricted and related machines such that, for any $I \in \mathcal{I}$,
\[    \max_{\sigma \in \LEX{I}} \frac{\csched[I, \sched]}{\copt(I)} = \Omega(\log \phi) \quad \mbox{and} \quad \max_{\sigma \in \LIST{I}} \frac{\csched[I, \sched]}{\copt(I)} = \Omega(\log \phi) \DOT \]
\end{theorem}


To prove this theorem, we first present a $\phi$-smooth instance and in
Algorithm~1, we implicitely give a permutation of the jobs such that list
scheduling using this permutation results in a schedule $\sched$ which
we will show is also lex-jump optimal. The schedule $\sched$ resembles
the worst case example constructed by Czumaj and
V\"{o}cking~\cite{Voecking:2007}: Machines are partitioned into classes
indexed by $0, 1, \ldots, r$. We will show that in $\sched$, each machine
in class $i$ has a load of approximately $i$, whereas the optimal
makespan is bounded by $3$. Hence, we can lower bound the performance
guarantee in the order of the number of classes. Whereas Czumaj and
V\"{o}cking needed $\Theta(\log m / \log \log m)$ classes, we only need
$\Theta(\log \phi)$ classes.

As scaling of all processing requirements does not change the approximation ratio, for sake of simplicity we do not consider probability densities $f_j \colon [0, 1] \rightarrow [0, \phi]$ but scaled densities $f_j' \colon [0, 2^{r+1}] \rightarrow [0, \phi/2^{r+1}]$ for an appropriate integer~$r$.

Let~$\phi \geq 4$ and consider an integer $r = \FLOOR{ \log_4 \phi } \geq 1$,
i.e., $\phi \geq 4^r = 2^{2r}$. The machines are partitioned into machine
classes~$M_k$ for $k = 0, \ldots, r$, such that machine class~$M_k$
contains~$r!/k!$ machines of speed~$2^k$. Also the jobs are partitioned into
job classes~$J_{\ell}$ for $\ell = 1, \ldots, r$ such that a job class
$J_\ell$ contains~$r!/(\ell-1)!$ jobs each having a processing
requirement uniformly drawn from $\left[ 2^\ell, 
2^\ell + 2^{r+1}/\phi \right) \subseteq (0, 2^{r+1}) $. Note that the
density of this instance is bounded by $\phi/2^{r+1}$ which is valid in
the variant of our model that we use in this subsection. The permutation
of the jobs is such that list scheduling constructs the schedule \sched in
the following way:

\begin{center}
{%
 \parbox{0.95\linewidth}%
 {%
  \textbf{Algorithm 1:}\\
   \quad 1. \quad \textbf{for} $k=1$ \textbf{to} $r$ \textbf{do}\\
   \quad 2. \quad\quad \textbf{for} $\ell=r$ \textbf{down to} $k$ \textbf{do}\\
   \quad 3. \quad\quad\quad Schedule $r! / \ell!$ arbitrary jobs of
class~$J_\ell$ according to list scheduling.\\
   \quad 4. \quad\quad \textbf{end for}\\
   \quad 5. \quad \textbf{end for}
 }%
}%
\end{center}
Note that for any job class~$J_\ell$ all $\ell\cdot r!/\ell! = r! /(\ell-1)!$ jobs have been scheduled. Let~$\sched$ be the resulting schedule. First, we show a key property of~$\sched$.

\begin{lemma}
\label{lemma.list.schedule}
For any index $\ell = 1, \ldots, r$ each machine in~$M_\ell$ is assigned
exactly~$\ell$ jobs of job class~$J_\ell$ and no other jobs. The machines in~$M_0$ remain empty.
\end{lemma}

\begin{proof}
Let~$\sched(k,\ell)$ denote the partial schedule after processing line~3
of iteration~$(k, \ell)$ of Algorithm~1. Within the $(k, \ell)^{th}$
iteration, we call a machine~$i \in M_\ell$ \emph{used} if a job of
class~$J_\ell$ has already been assigned to~$i$ during that iteration. Otherwise, we call machine~$i$ \emph{unused}. We show the two claims below inductively and simultaneously. The lemma then follows straightforwardly from the second claim since the last iteration is~$(r,r)$.

\begin{claim}
During iteration~$(k,\ell)$, $r!/\ell!$ jobs of class~$J_\ell$ are
assigned to $r!/\ell!$ distinct machines (i.e.\ all machines) of
class~$M_\ell$.
\end{claim}

\begin{claim}
\label{claim:nrjobs-partialschedule}
In the partial schedule~$\sched(k,\ell)$ each machine in class~$M_{\ell'}$ is assigned
\[ k'  = \left\{ \begin{array}{c@{\quad:\quad}l}
k & \ell' \geq \ell \COMMA \cr
\MIN{ k-1, \ell' } & \ell' < \ell \COMMA
\end{array} \right. \]
jobs of class~$J_{\ell'}$ and no other jobs.
\end{claim}

Figure~\ref{fig:list-lex-lb1} visualizes the partial schedule
$\sched(k,\ell)$. Machine~$i$ with speed $s_i = 2^i$ is a representative
for all machines in class~$M_i$. With~$\load{i}$ we refer to the current
load of machine~$i$ and with $\load[']{i}$ to the load of machine~$i$ at
the end of iteration $(k, k)$, i.e., in the partial schedule
$\sched(k,k)$. In phase $(k,\ell)$, $r! / \ell!$ jobs of size
roughly~$2^\ell$ are being assigned to the $r! / \ell!$ machines
in~$M_\ell$. All machines in~$M_{\ell'}$ for $\ell'>\ell$ just received
a job of roughly size~$2^{\ell'}$. All machines in~$M_{\ell'}$ for
$\ell' \in \SET{k, \ldots, \ell-1}$ will still receive a single job of
size roughly~$2^{\ell'}$ during iteration~$k$ of the outer loop. Figure~\ref{fig:list-lex-lb1} follows from the observations.

\begin{artclfig}
  \includegraphics[width=0.7\textwidth]{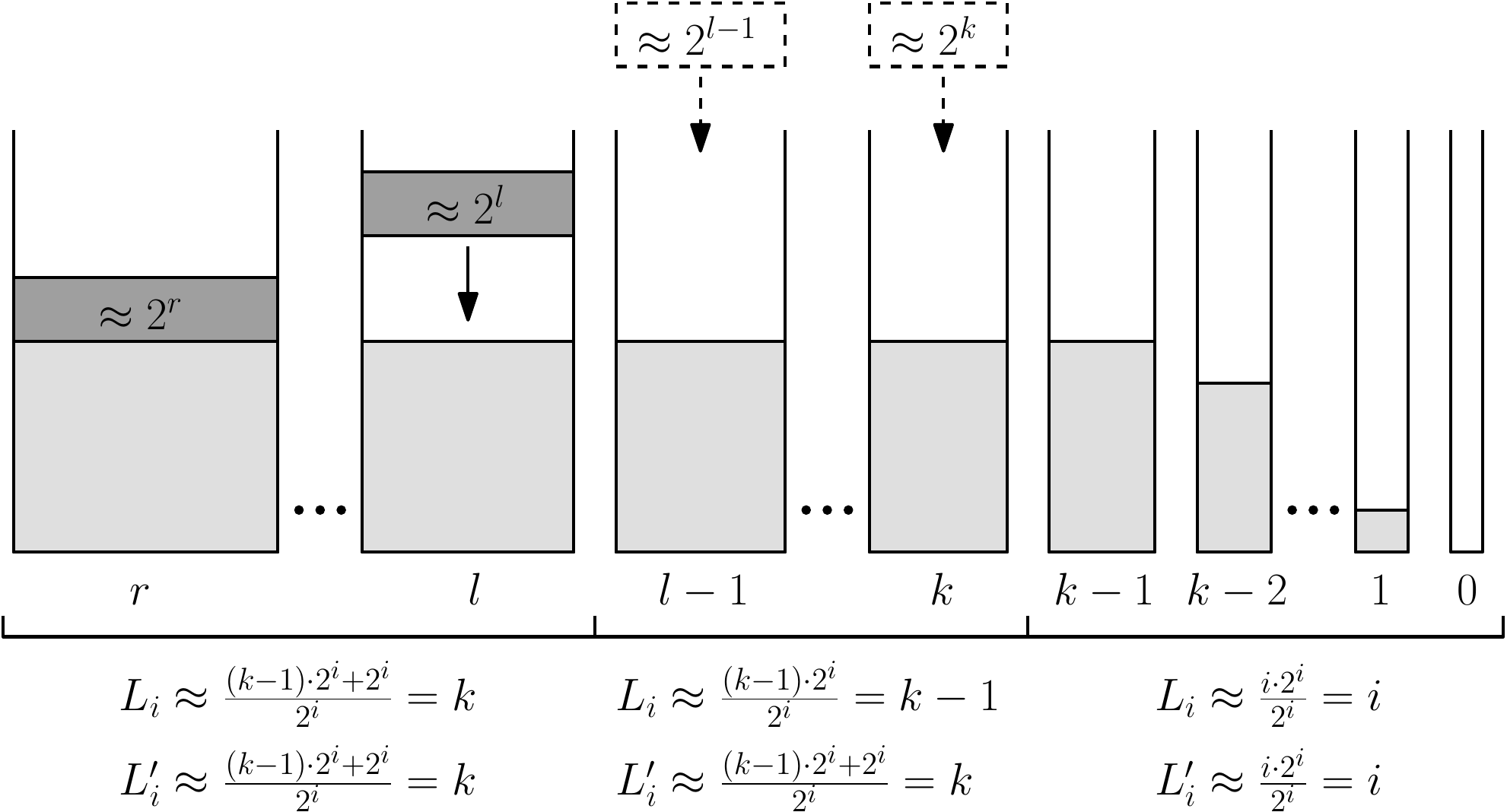}
  \caption{The partial schedule $\sched(k,\ell)$}
  \label{fig:list-lex-lb1}
\end{artclfig}

First, we validate the claims for the first iteration $(1,r)$. As only
$r!/r!=1$ job of class~$J_r$ has to be scheduled and since all machines
are still empty, the job will be scheduled on the fastest machine which
is the single machine in~$M_r$. Hence, both claims hold true for the
first iteration. Now, consider an arbitrary iteration~$(k,\ell)$ and
assume both claims hold true for all previous iterations. Consider a
job~$j \in J_\ell$ which needs to be assigned to a machine during
iteration~$(k, \ell)$. We show that job~$j$ will always be assigned to
an unused machine~$i \in M_\ell$. To see this, first note that the previous
iteration was either $(k,\ell+1)$ or $(k-1,k-1)$.

Let $i \in M_\ell$ be an unused machine. By the second claim, we know
that this machine carries~$k-1$ jobs of class~$J_\ell$. Consequently, we
can upper bound its load by
\[
  \load{i} +\frac{p_j}{s_i}
  < \frac{k \cdot (2^\ell + 2^{r+1}/\phi)}{2^\ell}
  = k + \frac{k}{\phi} \cdot 2^{r+1-\ell}
  \leq k + \frac{\ell}{2^{2r}} \cdot 2^{r+1-\ell}
  \leq k + \frac{1}{2^r} \COMMA
\]
where we used that~$k\leq \ell$,~$\phi \geq 2^{2r}$, and~$\ell/2^\ell
\leq 1/2$ for all integers~$\ell\geq 1$.

Consider a machine machine~$h$ which is either used (in that case let $\ell' =
\ell$) or in class~$M_{\ell'}$ for some $\ell' \in \SET{ \ell+1, \ldots,
r }$. By Claim~\ref{claim:nrjobs-partialschedule}, this machine carries~$k$ jobs of class~$J_ {\ell'}$ and thus
\[
  \load{h} +\frac{p_j}{s_h}
  \geq \frac{k \cdot 2^{\ell'} + 2^\ell}{2^{\ell'}}
  = k + 2^{\ell-\ell'}
  > k + \frac{1}{2^r}
  > \load{i} +\frac{p_j}{s_i} \DOT
\]

Finally, consider a machine~$h \in M_{\ell'}$ for some
$\ell' \in \SET{ 1, \ldots, \ell-1 }$. Again by
Claim~\ref{claim:nrjobs-partialschedule}, it carries~$\MIN{ k-1, \ell' }$ jobs
of class~$J_{\ell'}$ and thus
\begin{align*}
  \load{h} +\frac{p_j}{s_h}
  &\geq \frac{\MIN{ k-1, \ell' } \cdot 2^{\ell'} + 2^\ell}{2^{\ell'}}
  = \MIN{ k-1, \ell' } + 2^{\ell-\ell'} \cr
  &\geq (k - \MAX{ k-\ell', 1 }) + 2^{\MAX{ k-\ell', 1 }}
  \geq k + 1
  > \load{i} +\frac{p_j}{s_i} \COMMA
\end{align*}
where the second inequality follows from $\ell \geq \MAX{ k, \ell'+1 }$ and the third inequality follows from $2^i - i \geq 1$ for all positive integers~$i$.

With this complete case analysis we have shown that job~$j$ will be
assigned to an unused machine~$i \in M_\ell$. We conclude that during
iteration~$(k,\ell)$, each of the~$r!/\ell!$ jobs to be assigned will be
assigned to an unused machine in~$M_\ell$. Note that $|M_\ell|=r!/\ell!$, and hence for each job there always exists such an unused machine. The first claim and the second claim follow immediately.
\end{proof}

\begin{lemma}
Schedule~$\sched$ is lex-jump optimal.
\end{lemma}

\begin{proof}
It follows from Lemma~\ref{lemma.list.schedule} that the load of any
machine~$i \in M_\ell$ can be bounded by
\[
  \ell
  \leq \load{i}
  \leq \ell + \ell \cdot \frac{2^{r+1}}{2^\ell \phi}
  \leq \ell + \frac{\ell}{2^\ell} \cdot \frac{2^{r+1}}{2^{2r}}
  \leq \ell + \frac{1}{2} \cdot \frac{2^{r+1}}{2^{2r}}
  < \ell + 1 \DOT
\]
If a job~$j$ assigned to machine~$i \in M_\ell$ would jump to another
machine~$i' \in M_{\ell'}$, then
\[
  \load{i'} + \frac{p_j}{s_{i'}}
  \geq \ell' + \frac{2^\ell}{2^{\ell'}}
  = \ell - (\ell-\ell') +2^{\ell-\ell'}
  \geq \ell + 1 > \load{i} \COMMA
\]
where the last inequality follows from $2^k - k \geq 1$ for all integers~$k$. Thus, any job would be worse off by jumping to another machine, and hence schedule~$\sched$ is lex-jump optimal.
\end{proof}

We conclude this subsection by proving Theorem~\ref{thm.mainII.list.lex}.

\begin{proof}[\ProofText{Theorem~\ref{thm.mainII.list.lex}}]
We consider schedule~$\sched$ constructed above which is both a list
schedule and a lex-jump optimal schedule. By
Lemma~\ref{lemma.list.schedule} the load of the single machine in~$M_r$
is at least~$r$. Hence, $\csched \geq r$. Now, consider a
schedule~$\sched'$ in which each machine in~$M_\ell$ processes a single
job from job class~$J_{\ell+1}$, $\ell = 0, \ldots, r-1$. The single
machine in~$M_r$ remains empty. Then, the load of any machine~$i \in
M_\ell$ with job~$j$ assigned to it is bounded as follows:
\[ \load{i} = p_j / s_i \leq (2^{\ell+1}+ 2^{r+1}/\phi)/2^\ell \leq 2 + 2^{r+1}/(2^{2r} \cdot 2^1) = 2 + 2^{-r} < 3 \DOT \]
Hence, $\copt(I) \leq \csched[I, \sched'] < 3$ and the theorem follows: $\csched/\copt(I) \geq r/3 = \Omega(r) = \Omega(\log \phi)$.
\end{proof}

\section{Restricted Machines}
\label{sec:restricted}

In this section, we provide lower bound examples showing that the
worst-case performance guarantees for all variants of the restricted
machines are robust against random noise. Our lower bounds are in the
order of the worst-case bounds and hold in particular
for $\phi = 2$. In our lower bound constructions all processing
requirements are chosen uniformly at random from intervals of length~$1/2$.
This means that even with large perturbations the worst-case
lower bounds still apply.

\subsection{Jump Neighborhood on Restricted Machines}
\label{subsec:lb-jump-restricted}

Rutten et al.~\cite{Rutten:etal:2012} showed that the makespan of a
jump optimal schedule is at most a factor of $1/2 + \sqrt{m - 3/4}$ away from
the optimal makespan on restricted identical machines. 
On restricted related machines they
showed that the makespan of a jump optimal
schedule is not more than a factor of $1/2+\sqrt{(m-1)\cdot\smax+1/4}$
away from the makespan of an optimal schedule, assuming that $\smin=1$.
They provided two examples showing that the bound on identical machines
is tight and the one on related machines is tight up to a constant
factor.
We show that even on $\phi$-smooth
instances these bounds are tight up to a constant factor.
As in~\cite{Rutten:etal:2012}, we construct an example with two job
classes and three machine classes. The first machine class consists of
only one machine and this machine is the slowest among all machines. The
first class of jobs can only be scheduled on machines in the first two
classes, whereas the jobs in the second class are allowed on all
machines. To construct a bad example, we schedule all jobs in the first
class on the slowest machine and use the jobs of the second class to
fill the machines in the second machine class so that the schedule will
be jump optimal, with high probability. 

\begin{theorem}\label{thm:RestrictedJump}
For every~$\phi \geq 2$ there exists a class of $\phi$-smooth instances~$\mathcal{I}$ on
restricted related machines such that
\[ \El[I \sim \mathcal{I}]{\max_{\sched \in \JUMP{I}} \frac{\csched[I,
\sched]}{\copt(I)}} = \Omega \left( \sqrt{m \cdot \smax} \right), \]
assuming without loss of generality that $\smin = 1$.
\end{theorem}

\begin{proof}
It suffices to show the theorem for~$\phi=2$ and $m \geq 3$. W.l.o.g.\ we
assume~$\smin = 1$ and set $s := \smax/\smin = \smax$. Let~$z>2$ be an
arbitrary integer, let
\[
  m' = m-2 \geq 1 \COMMA \quad
  k' = \sqrt{\frac{m'}{s}} \leq \sqrt{m'} \COMMA \quad \text{and} \quad
  k = \CEIL{k'} \DOT 
\]
In the remainder we assume
that~$\sqrt{m's} \geq 17$. This is possible because we only want to derive
an asymptotic bound. We consider the following $\phi$-smooth
instance~$\mathcal{I}$. The set~$M$ of machines is partitioned into three
classes~$M_1$, $M_2$, and~$M_3$ such that
\[
  |M_1| = 1 \COMMA \quad
  |M_2| = k \COMMA \quad \text{and} \quad
  |M_3| = m' - (k-1) > m' - k' \geq 0 \DOT
\]
The machine in~$M_1$ has speed~$1$, the
machines in~$M_2$ have speed
\[
  s' = \max \SET{ 1, s \cdot \frac{k'}{k} } \in [1, s] \COMMA
\]
and the machines in~$M_3$ have speed~$s$. Let the set~$\jobset$ of jobs be
partitioned into two subsets~$\jobclass{1}$ and~$\jobclass{2}$, consisting of
\[
  |\jobclass{1}| = \FLOOR{2zsk'} \quad \text{and} \quad
  |\jobclass{2}| = \CEIL{32z s \cdot (m' - k')} \leq \CEIL{32zs \cdot |M_3|}
\]
jobs whose processing requirements are independently and uniformly drawn from
$[1/2, 1]$ and from $[0, 1/2]$, respectively. The jobs in~$\jobclass{1}$
are only allowed to be scheduled on the machines in~$M_1 \cup M_2$, whereas the
jobs in~$\jobclass{2}$ are allowed to be scheduled on any machine.

First, we construct a schedule~$\sched'$ to bound the optimal makespan: Use the
list scheduling algorithm to schedule all jobs in~$\jobclass{1}$ on the machines in~$M_2$, and all jobs
in~$\jobclass{2}$  on the machines in~$M_3$. Figure~\ref{fig:jump-restricted-opt} depicts schedule~$\sched'$. Machine~$i$ is a representative for all machines in class~$M_i$.

\begin{artclfig}\newcommand{\height}{12em}
  \includegraphics[height=\height]{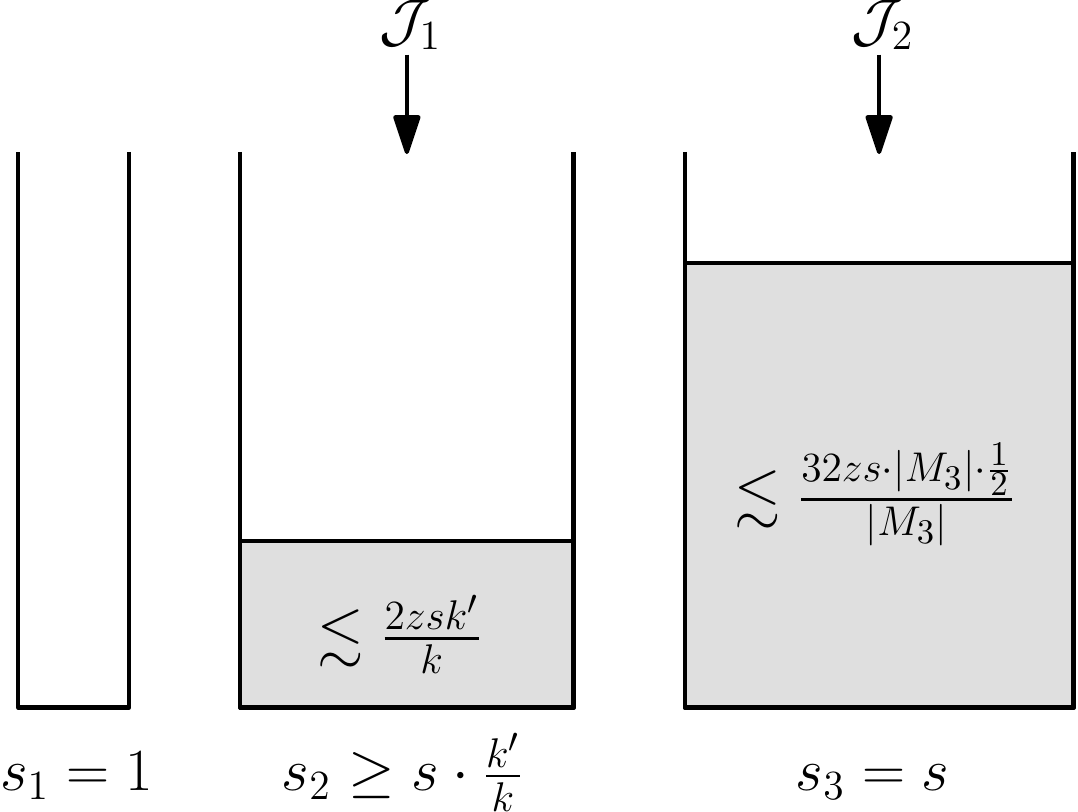}
  \caption{Schedule~$\sched'$}
  \label{fig:jump-restricted-opt}
\end{artclfig}

Along the same lines as in~\cite{graham:1966}, it follows that for all machines~$i \in M_2$
\[
  \load{i}
  \leq \frac{\frac{\sum_{j \in \jobclass{1}} p_j}{|M_2|} +
\max \limits_{j \in \jobclass{1}} p_j}{s'} 
 \leq \frac{\frac{|\jobclass{1}| \cdot 1}{|M_2|} + 1}{s'}
 \leq \frac{\frac{2zsk'}{k} + 1}{s'}
 \leq \frac{\frac{2zsk'}{k}}{s \cdot \frac{k'}{k}} + \frac{1}{1}
 = 2z + 1 \DOT
\]
Similarly,  for all machines~$i \in M_3$
\[
  \load{i}
  \leq \frac{\frac{\sum_{j \in \jobclass{2}} p_j}{|M_3|} +
\max \limits_{j \in \jobclass{2}} p_j}{s}
 \leq \frac{\frac{|\jobclass{2}| \cdot \frac{1}{2}}{|M_3|} +
\frac{1}{2}}{s}
 \leq \frac{\frac{32zs \cdot |M_3|}{2 \cdot |M_3|} +
1}{s}
 \leq 16z + 1 \DOT
\]

Hence, $\copt \leq \csched[\sched'] \leq 17z$. Before we proceed with
constructing a `bad' jump optimal schedule~$\sched$, we observe that
\begin{equation}
\label{eqref:UpperboundHats}
s' \leq 2s \cdot k'/k
\end{equation}
due to $1 \leq (\sqrt{m'} + 1)/k \leq 2\sqrt{m's}/k = 2s \cdot k'/k$.

We construct a jump optimal schedule~$\sched$ on the $\phi$-smooth
instance~$\mathcal{I}$ such that the corresponding makespan exceeds $zsk'$ with
high probability: Schedule all jobs in~$\jobclass{1}$ on the single machine
in~$M_1$. Then, $zsk' - 1 \leq \load{1} \leq 2zsk'$. Next, start assigning jobs
from~$\jobclass{2}$ to the machines in~$M_2$ according to the list scheduling
algorithm with an arbitrary job permutation, until
\begin{itemize}
\item[(a)] either~$\jobclass{2}$ becomes empty, or until
\item[(b)] $\load{i} \in \big[ \load{1} - \frac{1}{2s'}, \load{1} \big)$ for all~$i \in M_2$. If
there remain unscheduled jobs in~$\jobclass{2}$, then we assign them to
the machines in~$M_3$ using list scheduling.
\end{itemize}
Let~$Q = \sum_{j \in \jobclass{2}} p_j$ and let $\event$ denote the event that 
$Q > 4z(sk')^2$. If $\event$ occurs, then
\[
  \sum_{i \in M_2} s' \cdot \load{1}
  \leq |M_2| \cdot \left( 2s \cdot \frac{k'}{k} \right) \cdot 2zsk'
  = 4z(sk')^2
  < Q
\]
due to Inequality~\eqref{eqref:UpperboundHats}, i.e., the
algorithm will end up in case~(b) as~$p_j \leq 1/2$ for any job~$j \in
\jobclass{2}$. This shows that no machine~$i \in M_2$ is critical. Using the
same argument as for the analysis of~$\sched'$ we can show that the load of any
machine~$i \in M_3$ is bounded from above by $16z + 1 < 17z -1 \leq z \cdot
\sqrt{m' \cdot s} - 1 = zsk' - 1 \leq \load{1}$, i.e., the machine in~$M_1$ is the
unique critical machine. As each job on this machine has processing requirement
at least~$1/2$ and due to the property of the loads of the machines in~$M_2$ in
case~(b), schedule~$\sched$ is jump optimal and $\csched = \load{1} \geq zsk' -
1$.

It remains to determine the probability~$\Prob[\event]$. For this, note that
\begin{align*}
  \E[Q]
  &= \frac{|\jobclass{2}|}{4}
  \geq 8zs \cdot (m' - k')
  = 8zsm' \cdot \left( 1 - \frac{k'}{m'} \right) \cr
  &= 8zsm' \cdot \left( 1 - \frac{1}{\sqrt{m's}} \right)
  > 6zsm'
\end{align*}
as $\sqrt{m's} \geq 17$ by our initial assumption. On the other hand, $4z(sk')^2 = 4zsm'$.
Applying Hoeffding's Inequality~\cite{Hoeffding:1963} (see also Theorem~\ref{thm:app:hoeffding}), we obtain
\begin{align*}
  \Probl{\bar{\event}} 
  &= \Probl{Q \leq 4zsm'}
  \leq \Probl{Q - \Ee{Q} \leq -2zsm'} \cr
  &\leq \exp \left( - \frac{2 \cdot (2zsm')^2}{|\jobclass{2}| \cdot \left( \frac{1}{2}\right)^2} \right)
  \leq \exp \left( - \frac{32z^2 s^2 m'^2}{32zsm'+1} \right) \COMMA
\end{align*}
which becomes arbitrarily close to~$0$ when $z$ increases. Hence, for sufficiently large integers~$z$
\begin{align*}
  \El[I \sim \mathcal{I}]{\max_{\sched \in \JUMP{I}} \frac{\csched[I, \sched]}{\copt(I)}}
  &\geq  \El[I \sim \mathcal{I}]{ \left. \max_{\sched \in \JUMP{I}} \frac{\csched[I, \sched]}{\copt(I)} \right| \event} \cdot  \Probl[I \sim \mathcal{I}]{\event} \cr
  &\geq \frac{zsk' - 1}{17z} \cdot \frac{17}{18}
  \geq \frac{\sqrt{(m-2) \cdot \smax} - \frac{1}{z}}{18} \DOT \qedhere
\end{align*}
\end{proof}

\begin{cor}
For every~$\phi \geq 2$ there exists a class of $\phi$-smooth instances~$\mathcal{I}$ on
restricted identical machines such that
\[
  \El[I \sim \mathcal{I}]{\max_{\sched \in \JUMP{I}} \frac{\csched[I, \sched]}{\copt(I)}}
  = \Omega(\sqrt{m}) \DOT
\]
\end{cor}

\begin{remark}
In the proof of Theorem \ref{thm:RestrictedJump} we introduce an arbitrary integer $z$. We argue that there exists a sufficiently large value for $z$ such that the desired result follows. Choosing an even larger value for $z$ implies that the results above not only hold in expectation but also with high probability.
\end{remark}

\subsection{Lex-jump Optimal Schedules on Restricted Identical Machines}
\label{subsec:lb-lex-jump-restricted}
In this subsection, we show that there exist instances with~$\phi \geq 8$ such that the smoothed performance guarantee for lex-jump optimal schedules in the restricted setting is in the same order as the worst case performance guarantee.

As in Section~\ref{subsec:lb-list-lex-jump}, we construct an instance with
several job classes and machine classes and the loads of the machines
are gradually decreasing with increasing machine class.
By setting the sets~$\mathcal{M}_j$ of allowed machines equal to the union of
only one or two machine classes and choosing to schedule the jobs on the
\emph{wrong} machines, we can enforce that jobs cannot leave the machine class
on which they are scheduled in the lex-jump
optimal solution, whereas the optimal makespan is still small.

\begin{theorem}
\label{thm:rest-id-lj}
For every~$\phi \geq 8$ there exists a class of $\phi$-smooth instances~$\mathcal{I}$ on
restricted identical machines such that
\[
  \El[I \sim \mathcal{I}]{\max_{\sched \in \LEX{I}} \frac{\csched[I, \sched]}{\copt(I)}}
  = \Omega \left( \frac{\log m}{\log \log m} \right) \DOT
\]
\end{theorem}

First, we introduce the $\phi$-smooth instance~$\mathcal{I}$ for~$\phi \geq 8$. 
Given an integer~$k \geq 68$, consider the following recurrence formula:
\[
  a_0 = k^2 \COMMA \quad
  a_1 = k^3 \COMMA \quad \mbox{and} \quad
  a_h = \CEIL{ \left( \frac{a_{h-1}}{a_{h-2}} - \frac{7}{15} \right) \cdot a_{h-1} } \ \mbox{for} \ h \geq 2 \DOT
\]
Starting with $a_1/a_0 = k$, the fraction $a_h/a_{h-1}$ decreases with increasing index~$h$ until it is less or equal~$1$.
To see this, note that $a_h/a_{h-1} \geq 1$ implies that $a_{h-1} \geq a_{h-2}$. Therefore, we know that $a_h \geq a_{h-1} \geq \ldots \geq a_0 = k^2 > 15$. Furthermore, we can bound the ratio $a_h / a_{h-1}$ from above by $a_h / a_{h-1} \leq a_{h-1} / a_{h-2} - 7/15 + 1/a_{h-1} < a_{h-1} / a_{h-2} - 6/15 < a_{h-1} / a_{h-2}$.
Let~$z_k$ be the smallest integer~$h$ such that $a_h/a_{h-1} \leq 1$. Hence, $a_0, a_1, \ldots, a_{z_k-1}$ is a strictly increasing sequence. We will bound the number~$z_k$ from above later in the analysis.

We consider~$z_k$ job classes $\jobclass{1}, \ldots, \jobclass{z_k}$ and as many machine
classes $M_1, \ldots, M_{z_k}$. Each machine class~$M_h$ contains $m_h = a_{h-1}$ machines with speed~$1$. Each job class~$\jobclass{h}$ consists of two subclasses~$\jobclass{h}^A$ and~$\jobclass{h}^B$ of size~$a_h$ and of size $b_h = 17m_h$, respectively. The jobs in class~$\jobclass{h}^A$ are called type~$A$ jobs, have processing requirements independently and uniformly distributed in~$[7/8, 1]$, and can be processed on machines in~$M_h \cup M_{h+1}$. As a convention let $M_{z_k+1} = \emptyset$. Jobs in class~$\jobclass{h}^B$ are called type~$B$ jobs, have processing requirements independently and uniformly distributed in~$[0, 1/8]$, and can only be processed on machines in~$M_h$.

The schedule $\sched = \sched(I)$ for an instance~$I \in \mathcal{I}$ is obtained by scheduling the jobs in~$\jobclass{h}$ on the machines in~$M_h$ using LPT (longest processing time) scheduling, i.e., list scheduling with a list in which the jobs are ordered according to non-increasing processing requirements. Note that the LPT algorithm first schedules all type~A jobs and then all type~B jobs. Schedule~$\sched(I)$ is visualized in Figure~\ref{fig:lex-restricted-bad}. Machine~$h$ represents all machines in class~$M_h$.

\begin{artclfig}\newcommand{\height}{12em}
  \includegraphics[height=\height]{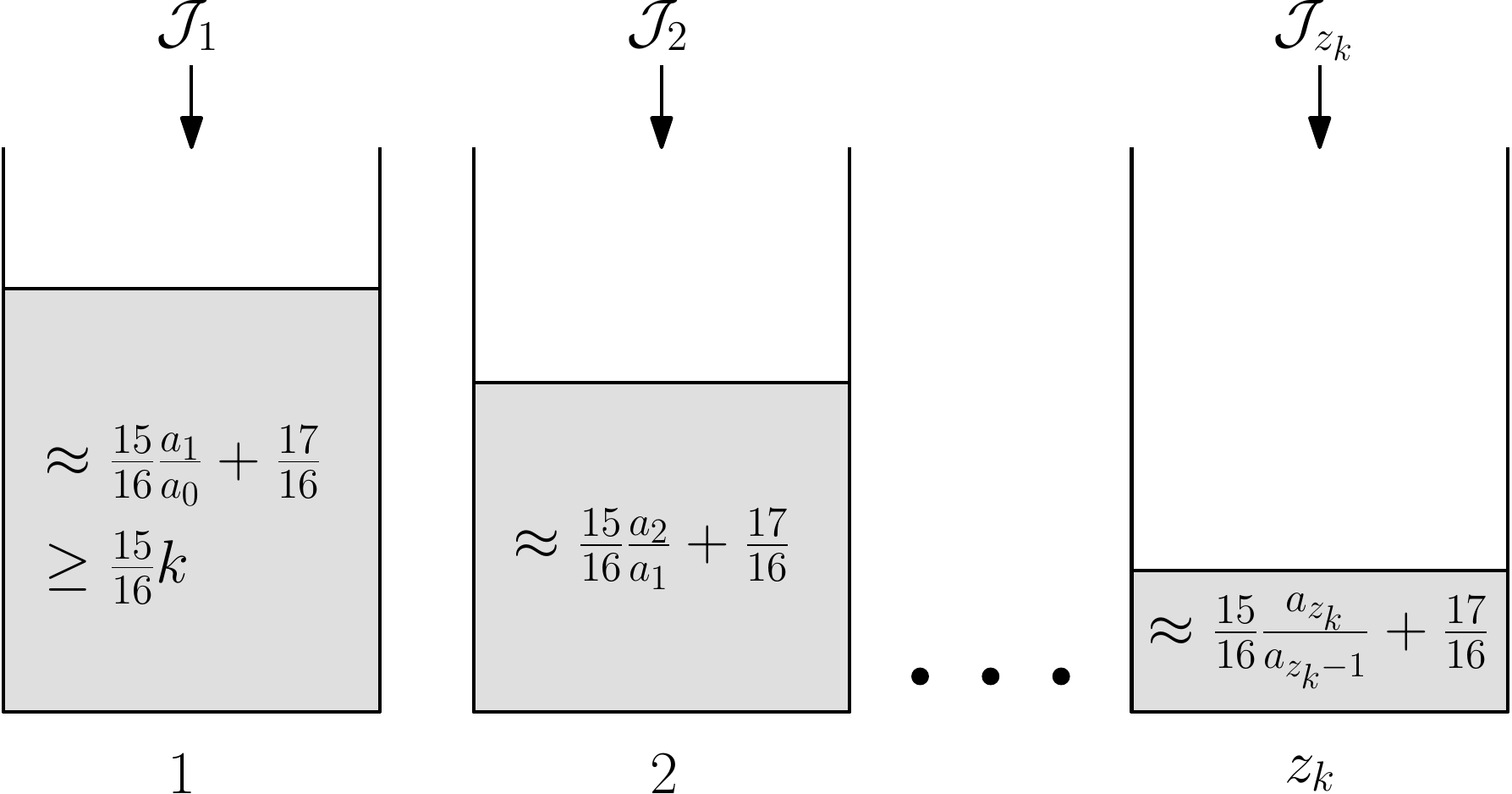}
  \caption{Schedule $\sched(I)$}
  \label{fig:lex-restricted-bad}
\end{artclfig}

We show that schedule~$\sched$ is lex-jump optimal with high probability. To be more specific, we show lex-jump optimality when the values~$Q_h^A = \sum_{j \in \jobclass{h}^A} p_j$ and~$Q_h^B = \sum_{j \in \jobclass{h}^B} p_j$ are close to their expectations for all $h = 1, \ldots, z_k$. Let~$\event_h^A$ and~$\event_h^B$ denote the events that
\[
  \big| Q_h^A - \El{Q_h^A} \big| \leq \frac{m_h}{16} \quad \mbox{and} \quad
  \big| Q_h^B - \El{Q_h^B} \big| \leq \frac{m_h}{32} \COMMA \quad \mbox{respectively} \DOT
\]
Moreover, let~$\event$ denote the event that the events~$\event_h^A$ and~$\event_h^B$ are simultaneously true for all $h = 1, \ldots, z_k$. By~$\bar{\event}_h^A$, $\bar{\event}_h^B$, and~$\bar{\event}$ we refer to the complement of~$\event_h^A$, $\event_h^B$, and~$\event$.

First, we analyze the sequence $a_0, a_1, \ldots, a_{z_k}$ to obtain bounds for the number~$z_k$ of machine and job classes and for the number~$m$ of machines.

\begin{lemma}
\label{lem:rest-id-lj:frac}
For any $h = 1, \ldots, z_k$ the following inequality holds:
\[
  \frac{a_h}{a_{h-1}} 
  \leq k - (h-1) \cdot \frac{2}{5} \DOT
\]
\end{lemma}

\begin{proof}
The claim is true for~$h = 1$. By definition of~$a_h$,
\[
  \frac{a_h}{a_{h-1}}
  \leq \frac{\left( \frac{a_{h-1}}{a_{h-2}} - \frac{7}{15} \right) \cdot a_{h-1} + 1}{a_{h-1}}
  \leq \frac{a_{h-1}}{a_{h-2}} - \frac{6}{15}
  = \frac{a_{h-1}}{a_{h-2}} - \frac{2}{5}
\]
for any $h = 2, \ldots, z_k$ as $a_{h-1} \geq a_0 = k^2 \geq 15$. The claim follows by induction.
\end{proof}

Now, we can bound the number~$z_k$ of job classes.

\begin{cor}
\label{cor:rest-id-lj:zk}
The number~$z_k$ of machine classes and job classes is bounded by~$5k/2$.
\end{cor}

\begin{proof}
Applying Lemma~\ref{lem:rest-id-lj:frac} for $h = z_k-1$ we obtain
\[
  1 < \frac{a_{z_k-1}}{a_{z_k-2}} \leq k - (z_k-2) \cdot \frac{2}{5} \DOT
\]
Hence,
\[
  z_k < (k-1) \cdot \frac{5}{2} + 2 < \frac{5k}{2} \DOT \qedhere
\]
\end{proof}

\begin{lemma}
\label{lem:rest-id-lj:m}
The number~$m$ of machines is bounded by $\Gamma(k'+3)$ where~$\Gamma$ denotes the gamma function and where $k' = \CEIL{5k/2}$.
\end{lemma}

\begin{proof}
By induction we show that
\[
  a_h \leq k^2 \cdot \left( \frac{2}{5} \right)^h \cdot \frac{k'!}{(k'-h)!}
\]
for any $h = 0, \ldots, z_k-1$. Note that $z_k \leq 5k/2 \leq k'$ due to Corollary~\ref{cor:rest-id-lj:zk}. For~$h=0$ the claim holds since $a_0 = k^2$. For~$h \geq 1$ we apply Lemma~\ref{lem:rest-id-lj:frac} to get
\[
  \frac{a_h}{a_{h-1}}
  \leq k - (h-1) \cdot \frac{2}{5}
  \leq \frac{2}{5} \cdot (k'-(h-1)) \DOT
\]
The induction hypothesis for~$a_{h-1}$ yields
\[
  a_h
  \leq \frac{2}{5} \cdot (k'-(h-1)) \cdot k^2 \cdot \left( \frac{2}{5} \right)^{h-1} \cdot \frac{k'!}{(k'-(h-1))!}
  = k^2 \cdot \left( \frac{2}{5} \right)^h \cdot \frac{k'!}{(k'-h)!} \DOT
\]
Recalling $m_h = a_{h-1}$ we can bound the number~$m$ of machines by using
\[
  \frac{m}{k^2}
  = \sum \limits_{h=1}^{z_k} \frac{m_h}{k^2}
  = \sum \limits_{h=0}^{z_k-1} \frac{a_h}{k^2}
  \leq \sum \limits_{h=0}^{z_k-1} \frac{k'!}{(k'-h)!}
  \leq k'! \cdot e \DOT
\]
Hence, $m \leq e \cdot k^2 \cdot k'! \leq (k'+2)! = \Gamma(k'+3)$.
\end{proof}

\begin{lemma}
\label{lem:rest-id-lj:succ-prob}
Event~$\bar{\event}$ occurs with probability at most~$10k \cdot \exp(-k/2)$.
\end{lemma}

\begin{proof}
We bound the probability for the events~$\bar{\event}_h^A$ and~$\bar{\event}_h^B$ to occur. Recalling $m_h = a_{h-1} \geq a_0 = k^2$, $a_h \leq k \cdot a_{h-1}$ (see Lemma~\ref{lem:rest-id-lj:frac}), $b_h = 17m_h$, and $k \geq 68$ we obtain
\begin{align*}
  \Probl{\bar{\event}_h^A}
  &= \Probl{\Big| Q_h^A - \Ee{Q_h^A} \Big| > \frac{m_h}{16}}
  \leq 2 \exp \left( - \frac{2 \left( \frac{m_h}{16} \right)^2}{a_h \cdot \left( \frac{1}{8} \right)^2} \right) \cr
  &= 2 \exp \left( - \frac{a_{h-1}}{a_h} \cdot \frac{a_{h-1}}{2} \right)
  \leq 2 \exp \left( -\frac{a_{h-1}}{2k} \right)
  \leq 2 \exp \left(-\frac{k}{2} \right)
\end{align*}
and
\begin{align*}
  \Probl{\bar{\event}_h^B}
  &= \Probl{\Big| Q_h^B - \Ee{Q_h^B} \Big| > \frac{m_h}{32}}
  \leq 2 \exp \left( - \frac{2 \left( \frac{m_h}{32} \right)^2}{b_h \cdot \left( \frac{1}{8} \right)^2} \right) \cr
  &= 2 \exp \left( - \frac{m_h}{17m_h} \cdot \frac{a_{h-1}}{8} \right)
  \leq 2 \exp \left( -\frac{k^2}{136} \right)
  \leq 2 \exp \left( - \frac{k}{2} \right) \DOT
\end{align*}
Each of the first inequalities stems from Hoeffding's bound~\cite{Hoeffding:1963} (see also Theorem~\ref{thm:app:hoeffding}). A union bound yields
\[
  \Probl{\bar{\event}}
  = \Probl{\bigcup_{h=1}^{z_k} \big( \bar{\event}_h^A \cup \bar{\event}_h^B \big)}
  \leq 2z_k \cdot 2\exp \left( -\frac{k}{2} \right)
  \leq 10k \cdot \exp \left( -\frac{k}{2} \right)
\]
due to Corollary~\ref{cor:rest-id-lj:zk}.
\end{proof}

As event~$\event$ occurs with high probability and as
\[
  \El[I \sim \mathcal{I}]{\max_{\sched \in \LEX{I}} \frac{\csched[I, \sched]}{\copt(I)}}
  \geq \El[I \sim \mathcal{I}]{\left. \max_{\sched \in \LEX{I}} \frac{\csched[I, \sched]}{\copt(I)} \right| \event} \cdot \Probl[I \sim \mathcal{I}]{\event} \COMMA
\]
to prove Theorem~\ref{thm:rest-id-lj} it suffices to bound the expected value conditioned on event~$\event$ by $\Omega \big( \frac{\log m}{\log \log m} \big)$. Therefore, in the remainder of this section we assume that event~$\event$ happens.

\begin{lemma}
\label{lemma:rest-id-lj:load-dif}
The loads of the machines within the same class differ only slightly. In particular, $|\load{i} - \load{i'}| \leq 1/8$ for any machines $i, i' \in M_h$.
\end{lemma}

\begin{proof}
Suppose to the contrary that there exist two machines $i, i' \in M_h$
such that $\load{i} - \load{i'} > 1/8$. Recall that according to the LPT rule all type~$A$ jobs will be assigned to the machines before the type~$B$ jobs are assigned. After all type~$A$ jobs  have been assigned to the machines in~$M_h$, the difference in load between any two machines in~$M_h$ is at most~$1$ since~$p_j \leq 1$ for all jobs~$j$.

Since the processing time of all type~$B$ jobs is bounded by~$1/8$, $\load{i} - \load{i'} > 1/8$ implies that  no type~$B$ job is assigned to machine~$i$ nor to any machine that has load at least~$\load{i}$. Hence, all type~$B$ jobs are assigned to the machines that have load less than~$\load{i}$. Note that there are at most $m_h-1$ such machines.

As the difference in load between machine~$i$ and any other machine in~$M_h$ is
at most~$1$, the total amount of processing requirements of type~$B$ jobs in
class~$M_h$ is bounded by $Q_h^B \leq (m_h - 1) \cdot 1 < 17m_h/16 - m_h/32 = \E[Q_h^B] - m_h/32$
contradicting the assumption that event~$\event_h^B$ holds.
\end{proof}

\begin{lemma}
\label{lemma:rest-id-lj:load}
For any machine~$i \in M_h$ the inequality
\[
  \left| \load{i} - \frac{1}{m_h} \left( \El{Q_h^A} + \El{Q_h^B} \right) \right| \leq \frac{7}{32}
\]
holds, i.e., the load of machine~$i$ is close to the expected average machine load in class~$M_h$.
\end{lemma}

\begin{proof}
By applying the triangle inequality we obtain
\begin{align*}
  \left| \load{i} - \frac{\El{Q_h^A} + \El{Q_h^B}}{m_h} \right|
  &\leq \left| \load{i} - \frac{Q_h^A + Q_h^B}{m_h} \right| + \frac{\left| Q_h^A - \El{Q_h^A} \right|}{m_h} +  \frac{\left| Q_h^B - \El{Q_h^B} \right|}{m_h} \cr
  &\leq \left| \load{i} - \frac{\sum \limits_{i' \in M_h} \load{i'}}{|M_h|} \right| + \frac{1}{16} + \frac{1}{32}
  \leq \frac{7}{32} \COMMA
\end{align*}
where the second inequality holds since~$\event_h^A$ and~$\event_h^B$ are true. The last inequality is due to Lemma~\ref{lemma:rest-id-lj:load-dif}.
\end{proof}

\begin{lemma}
\label{lem:rest-id-lj:lexjumpopt2}
Schedule~$\sched$ is lex-jump optimal.
\end{lemma}

\begin{proof}
We need to show that $\load{i'} + p_j \geq \load{i}$ holds for any machine~$i \in M_h$, any job~$j \in
\jobset[i]$, and any machine~$i' \in \allow{j}$. Let~$i \in M_h$ be an arbitrary machine. First, consider the last job~$j$ that has been assigned to~$i$. Then, $\load{i'} + p_j \geq \load{i}$ for any machine~$i' \in M_h$ as this job was assigned to machine~$i$ by list scheduling. Furthermore, job~$j$ is a smallest job on machine~$i$ due to the LPT rule. Hence, $\load{i'} + p_{j'} \geq \load{i}$ for any machine~$i' \in M_h$ and any job~$j' \in \jobset[i]$ assigned to machine~$i$.

For type~$B$ jobs on machine~$i$ the set of allowed machines equals~$M_h$. It just remains to show that $\load{i'} + p_j \geq \load{i}$ for any machine~$i' \in M_{h+1}$ and any type~$A$ job $j \in \jobset[i]$ with $i \in M_h$. Recalling $a_h = \CEIL{(a_{h-1}/a_{h-2}) - 7/15) \cdot a_{h-1}}$ for~$h \geq 2$, $m_h = a_{h-1}$, and $b_h/m_h = 17$ we observe that
\begin{align*}
  \frac{\El{Q_{h+1}^A} + \El{Q_{h+1}^B}}{m_{h+1}}
  &= \frac{\frac{15}{16} a_{h+1} + \frac{1}{16} b_{h+1}}{m_{h+1}}
  = \frac{15}{16} \cdot \frac{a_{h+1}}{a_h} + \frac{1}{16} \cdot \frac{b_{h+1}}{m_{h+1}} \cr
  &\geq \frac{15}{16} \cdot \left( \frac{a_h}{a_{h-1}} - \frac{7}{15} \right) + \frac{1}{16} \cdot \frac{b_h}{m_h} \cr
  &= \frac{\El{Q_h^A} + \El{Q_h^B}}{m_h} - \frac{7}{16}
\end{align*}
for any $h = 1, \ldots, z_k - 1$. This implies
\begin{align*}
  \load{i'} + p_j
  &\geq \frac{\El{Q_{h+1}^A} + \El{Q_{h+1}^B}}{m_{h+1}} - \frac{7}{32} + \frac{7}{8} \cr
  &\geq \frac{\El{Q_h^A} + \El{Q_h^B}}{m_h} - \frac{7}{16} + \frac{21}{32} \cr
  &= \frac{\El{Q_h^A} + \El{Q_h^B}}{m_h}  + \frac{7}{32}
  \geq \load{i} \COMMA
\end{align*}
where the first and the last inequality are due to Lemma~\ref{lemma:rest-id-lj:load}.
\end{proof}

Finally, we can prove Theorem~\ref{thm:rest-id-lj}.

\begin{proof}[\ProofText{Theorem~\ref{thm:rest-id-lj}}]
As mentioned before, due to Lemma~\ref{lem:rest-id-lj:succ-prob} it suffices to bound the expected value conditioned on event~$\event$. If event~$\event$ holds, then schedule~$\sched = \sched(I)$ is lex-jump optimal (see Lemma~\ref{lem:rest-id-lj:lexjumpopt2}), i.e., $\sched \in \LEX{I}$, and has makespan
\begin{align*}
  \cmax
  &\geq \max \limits_{i \in M_1} \load{i}
  \geq \frac{Q_1^A + Q_1^B}{m_1}
  \geq \frac{\Ee{Q_1^A} + \Ee{Q_1^B}}{m_1} - \frac{\frac{m_1}{16} + \frac{m_1}{32}}{m_1} \cr
 &= \frac{\frac{15}{16}k^3 + \frac{1}{16} \cdot 17k^2}{k^2} - \frac{3}{32} \geq \frac{15}{16}k \COMMA
\end{align*}
where the third inequality is due to the occurrence of~$\event_1^A$ and~$\event_1^B$. Now, consider the following schedule~$\sched'$:
\begin{itemize}

  \item For $h = 1, \ldots, z_k-1$ spread the jobs of class~$\jobclass{h}^A$ evenly among the machines in class~$M_{h+1}$. As $|\jobclass{h}^A| = a_h = m_{h+1} = |M_{h+1}|$, each machine is assigned exactly one type~$A$ job.

  \item Spread the jobs of class~$\jobclass{z_k}^A$ evenly among the machines in class~$M_{z_k}$. As $|\jobclass{z_k}^A| = a_{z_k} \leq a_{z_k-1} = m_{z_k} = |M_{z_k}|$, each machine is assigned at most one type~$A$ job.

  \item For $h = 1, \ldots, z_k$ spread the jobs of class~$\jobclass{h}^B$ evenly among the machines in class~$M_h$. As $|\jobclass{h}^B| = 17m_h = 17 \cdot |M_h|$, each machine is assigned exactly $17$~type~$B$ jobs.

\end{itemize}
Note that with `evenly' we refer to the number of jobs on each machine and not to the load. Figure~\ref{fig:lex-restricted-opt} shows schedule~$\sched'$ where each machine~$h$ is a representative for all machines in class~$M_h$.

\begin{artclfig}\newcommand{\height}{12em}
  \includegraphics[height=\height]{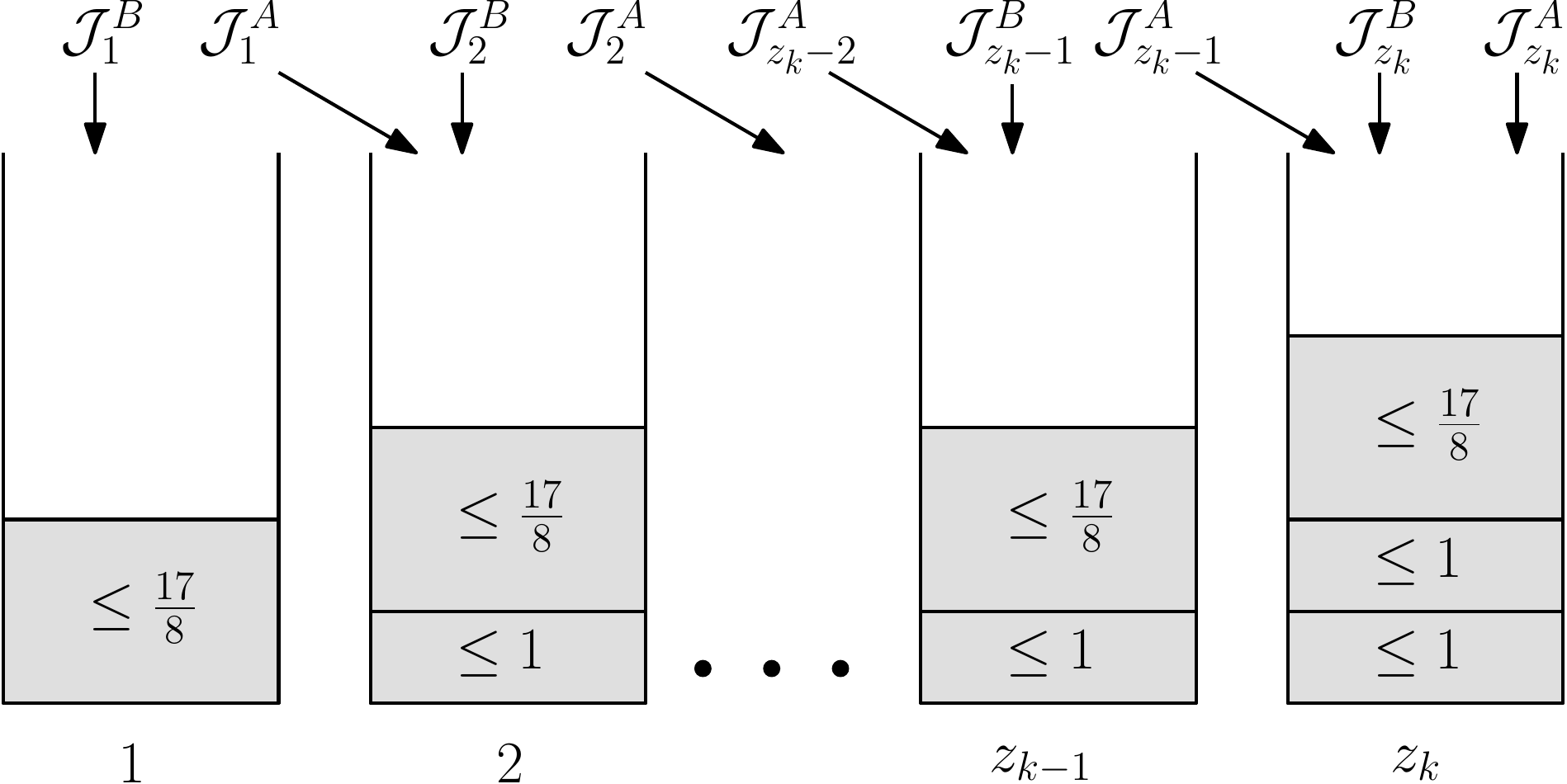}
  \caption{Schedule $\sched'$}
  \label{fig:lex-restricted-opt}
\end{artclfig}

As each machine contains at most $2$~type~$A$ jobs and $17$~type~$B$ jobs, the makespan of schedule~$\sched'$ and hence~$\copt$ is bounded by $2 \cdot 1 + 17 \cdot 1/8 \leq 5$. This implies $\cmax(\sched)/\copt \geq 3k/16 = \Omega(\Gamma^{-1}(m))$ due to Lemma~\ref{lem:rest-id-lj:m}. Hence,
\[ \El[I \sim \mathcal{I}]{\left. \max_{\sched \in \LEX{I}} \frac{\csched[I, \sched]}{\copt(I)} \right| \event} \geq \El[I \sim \mathcal{I}]{\left. \frac{\csched[I, \sched(I)]}{\copt(I)} \right| \event} = \Omega\left(\frac{\log m}{\log \log m}\right) \DOT \qedhere \]
\end{proof}

\begin{remark}
The worst case upper bound on the performance guarantee for
lex-jump optimal schedules on restricted related machines is
$O\left(\frac{\log S}{\log \log S} \right)$, where $S = \sum_{i}
s_i /s_m$~\cite{Rutten:etal:2012}.
As for identical machines $S = m$, i.e., each machine has speed~$1$, the
upper bound matches the lower bound of Theorem~\ref{thm:rest-id-lj} up
to a constant factor and smoothing does also not improve the performance
guarantee for the worst lex-jump optimal schedules on restricted related
machines.

Lemma \ref{lem:rest-id-lj:succ-prob} established that $\event$ occurs with high probability. Hence, if we choose $k$ suitably large, the stated results not only hold in expectation, but also with high probability.
\end{remark}

\section{Concluding Remarks}
\label{sec:concluding}

We have proven that the lower bounds for all scheduling variants with restricted machines
are rather robust against random noise, not only in expectation but even with high probability. We have also shown that the situation looks much
better for unrestricted machines where we obtained performance guarantees of~$\Theta(\phi)$ and~$\Theta(\log \phi )$ for the jump and lex-jump algorithm, respectively.
The latter bound also holds for the price of anarchy of routing on parallel links and
for the list scheduling algorithm, even when the order in which the jobs are presented to the algorithm can be chosen by the adversary when the realization of the processing times are known.

There are several interesting directions of research and we view our results
only as a first step towards fully understanding local search and greedy algorithms
in the framework of smoothed analysis. For example, we have only
perturbed the processing requirements, and it might be the case that the worst-case
bounds for the restricted scheduling variants break down if also the sets~$\allow{j}$
are to some degree random. In general it would be interesting to study different
perturbation models where the sets~$\allow{j}$ and/or the speeds~$s_i$ are perturbed. 
Lemma~\ref{lemma.exponentially.decreasing.speeds} and Corollary~\ref{corol.many.small.jobs} indicate that there need to exist many machines having exponentially small speeds. We conjecture that if speeds are being smoothed, then the smoothed performance guarantee of near list schedules on restricted related machines is $\Theta(\log \phi)$ as well.

Another interesting question is the following:
since we do not know which local optimum is reached, we have always looked at
the worst local optimum. It might, however, be the case that the local optima
reached in practice are better than the worst local optimum. It would be interesting
to study the quality of the local optimum reached under some reasonable assumptions on how
exactly the local search algorithms work. An extension in this direction would be 
to analyze the quality of coordination mechanisms under smoothing.

\section*{Acknowledgments}
We thank three anonymous referees for their valuable comments and suggestions 
that helped to improve the writing of the paper.


\newpage
\appendix

\section{Table of notation}
\label{sec:appendix-table}

In the table below, the notation used in this paper is summarized.
\vspace{1em}

\begin{tabular}{|ll|} \hline
$J$ & set of jobs $1, \ldots, n$ \\
$M$ & set of machines $1, \ldots, m$ \\
$p_j$ & processing requirement of job~$j$ \\
$s_i$ & speed of machine~$i$ \\
${\cal M}_j$ & set of machines on which job $j$ can be scheduled \\
$\smax$ & maximum speed of the machines \\
$\smin = 1$ & minimum speed of the machines;  \\
& by scaling we assume w.l.o.g. it to be $1$. \\
$\copt$ & optimal makespan \\
$\cmax(\sched)$ & makespan of schedule $\sched$ \\
$\jobset[i](\sched)$ & set of jobs scheduling on machine $i$ in schedule
$\sched$ \\
$\load{i}(\sched)$ & $ = \sum_{j \in \jobset[i](\sched)} p_j / s_i$ \\
 & load of
machine $i$ in schedule $\sched$. \\
$\jobset[i,j](\sched)$ & $ = \jobset[i](\sched) \cap \{ 1, \ldots, j \}$ \\
$j_i^t$ & $= \MIN{j \WHERE \sum_{\ell \in \jobset[i,j](\sched)} p_{\ell} / s_i \geq t \cdot \copt}$ \\
$\jobset[i,\geq t](\sched)$ & $= \jobset[i,j_i^t](\sched)$ \\
$c$ & $ = \left\lfloor \frac{\cmax(\sched)}{\copt} \right\rfloor - 1$ \\
$i_k$ & $= \MAX{i \in M \WHERE \load{i'} \geq k \cdot \copt \, \forall \, i' \leq i }$, \\
& assuming $s_1 \geq s_2 \geq \ldots \geq s_m$  \\
$H_k$ & $= \{1, \ldots, i_k\}$ \\
$R_k$ & $= H_k \setminus H_{k+1}$ \text{ for $k=0,1,\ldots, c-1$} \\
$R_c$ & $= H_c$.
\\ \hline
\end{tabular}

\section{Hoeffding's bound}

On several occasions in this paper we use Hoeffding's
bound~\cite{Hoeffding:1963} to bound tail probabilities. For
completeness, we state the bound in the following theorem.

\begin{theorem}
\label{thm:app:hoeffding}
Let $X_1, \ldots, X_n$ be independent random variables.
Define $X := \sum_{j=1}^n X_j$ and $\mu = \Ee{X}$. If each $X_j \in [a_j,b_j]$
for some constants $a_j$ and $b_j$, $j = 1,\ldots,n$, then for any $t >
0$
\begin{align*}
\Probl{X \leq \Ee{X} - t}
&\leq \exp\left( \frac{-2t}{\sum_j (b_j - a_j)^2} \right), \quad \text{
and, } \\
\Probl{X \geq \Ee{X} + t}
&\leq \exp\left( \frac{-2t}{\sum_j (b_j - a_j)^2} \right). 
\end{align*}
\end{theorem}
\end{document}